\documentclass[runningheads]{llncs}
\usepackage[T1]{fontenc}
\usepackage{graphicx}
\usepackage{multirow}
\usepackage{charter}

\usepackage[backgroundcolor=SkyBlue,linecolor=SkyBlue,obeyFinal]{todonotes}

\makeatletter
\providecommand\@dotsep{5}
\def\listtodoname{}
\def\listoftodos{\@starttoc{tdo}\listtodoname}
\makeatother

\newcounter{nmcomment}

\usepackage{parskip}
\usepackage{hyperref}
\usepackage{comment}
\excludecomment{shortver}
\includecomment{longver}

\usepackage{wrapfig}
\usepackage{tikz}
\usetikzlibrary{positioning}
\usetikzlibrary{arrows}
\usetikzlibrary{decorations.pathmorphing}
\usetikzlibrary{decorations.markings}

\usepackage{latexsym,amssymb,amsmath,mathtools,eulervm,xspace}

\begin{document}
\title{Linear Space Data Structures for Finite Groups with Constant Query-time}

\author{Bireswar Das\inst{1} \and Anant Kumar\inst{1} \and Shivdutt Sharma\inst{2} \and Dhara Thakkar \inst{1} \thanks{Funded by CSIR-UGC NET JRF Fellowship.}}
\authorrunning{B. Das, A. Kumar, S. Sharma, and D. Thakkar}
\institute{Indian Institute of Technology Gandhinagar, Gandhinagar \and
Indian Institute of Information Technology, Una, India
\email{\{bireswar,kumar\_anant,thakkar\_dhara\}@iitgn.ac.in}\\
\email{shiv@iiitu.ac.in}
}
%shiv@iiitu.ac.in
\maketitle              % typeset the header of the contribution
\begin{abstract}
A finite group of order $n$ can be represented by its Cayley table. In the word-RAM model the Cayley table of a group of order $n$ can be stored using $O(n^2)$ words and can be used to answer a multiplication query in constant time. It is interesting to ask if we can design a data structure to store a group of order $n$ that uses $o(n^2)$ space but can still answer a multiplication query in constant time. 

We design a constant query-time data structure that can store any finite group using  $O(n)$ words where $n$ is the order of the group.

Farzan and Munro (ISSAC 2006) gave an information theoretic lower bound of $\Omega(n)$ on the number of words to store a group of order $n$. Since our data structure achieves this lower bound and answers queries in constant time, it is optimal in both space usage and query-time.

A crucial step in the process is essentially to design linear space and constant query-time data structures for nonabelian simple groups. The data structures for nonableian simple groups are designed using a lemma that we prove using the Classification Theorem for Finite Simple Groups (CFSG).


{\bf Related Version:} A preliminary version of this article appeared in the proceedings of the 39th International Symposium on Theoretical Aspects of Computer Science (STACS 2022).

%\keywords{First keyword  \and Second keyword \and Another keyword.}
\end{abstract}

\section{Introduction}\label{Introduction}

The Cayley table of a group of order $n$ is a two dimensional table whose $(i,j)$th entry is the product of the $i$th and $j$th element of the group. In the word-RAM model while it takes $O(n^2)$ words to store the Cayley table of a group of order $n$, a multiplication query can be answered in constant time by accessing the appropriate location of the table.
 
For many computational problems in group theory the input group is given by its Cayley table. Some of these problems include the minimum generating set problem, various problems in property testing, the group factoring problem, and the group isomorphism problem \cite{Kumar-Ravi,Arvind-Jacobo,Kayal-Neeraj,Miller}. Among these, the group isomorphism problem is probably the most prominent one because of its unresolved complexity status despite years of extensive research \cite{Babai,Le-Gall,Jacobo-Arvind,Babai-1,Jayalal-Sarma,Kavitha-T}.

The Cayley table is very fast in terms of query processing but it takes quadratic space to store a group. It is interesting to ask if we can design a data structure for finite groups using $o(n^2)$ space\footnote{In this paper we use the word-RAM model. The space used by a data structure or an algorithm refers to the number of words used by them.} which can still answer multiplication query in constant time. We note that while quasigroups, and semigroups can also be stored using their Cayley tables, it is not possible to store quasigroups, and semigroups using $o(n^2)$ space. This is simply because the numbers of quasigroups, and semigroups are too large \cite{Lint-Wilson,Daniel-Rothschild} and the information theoretic lower bound is $\Omega(n^2\log n)$ bits or $\Omega(n^2)$ words. 

Das et al. \cite{JCSS} showed that for any finite group $G$ of order $n$ and for any $\delta\in[1/\log{n}, 1]$, a data structure can be constructed for $G$ that  uses $O(n^{1+\delta}/\delta)$ space and answers a multiplication query in time $O(1/\delta)$. Their result implies that there exist constant query-time data structures for finite groups of order $n$ that use $O(n^{1.01})$ space. However, the result cannot be used to design a constant query-time data structure even if we are allowed to use $\Theta(n.polylog(n))$ space.

In this paper we design constant query-time data structures for finite group that can be stored using $O(n)$  words where $n$ is the order of the group.  An information theoretic argument by Farzan and Munro shows that a lower bound to store a group of order $n$ is $\Omega(n\log n)$ bits or $\Omega(n)$ words \cite{Farzan}. Our data structure is optimal in the sense that it achieves the lower bound. A data structure that achieves the optimum information theoretic lower bound asymptotically is known as \emph{a compact data structure}. Therefore our data structure is a constant query-time compact data structure for finite groups. We note that compact query-time data structures were designed for some restricted classes of groups such as abelian groups and Dedekind groups \cite{WALCOM}.

In the process of designing the data structure we first prove two results, which we call \emph{extension theorems}, on the construction of data structures for a group when we already have a data structure for a subgroup of the given group. The extra space used by the newly constructed data structure depends on the index of the subgroup in one of the results and the structure\footnote{The subgroup needs to be normal and quotient needs to be cyclic.} of the subgroup in the other result. This indicates that finding suitable subgroups of a group might be useful.

The Jordan-H\"older theorem provides us with a supply of subgroups in the form of composition series. In our process we try to pick some suitable subgroups that are elements of the composition series of the given group. However, picking suitable groups is not always possible. This happens, as we will see in Section~\ref{Compact Data structure}, when  there is a ``large'' composition factor sitting in a certain position of the composition series. The composition factors are simple groups. In a sense the hard cases for constructing the data structure are for the simple groups.

Simple groups are sometimes considered as the building blocks for finite groups. The Classification Theorem for Finite Simple Groups (CFSG) is one of the most important theorems in group theory. Informally, this theorem classifies the finite simple groups into cyclic groups, alternating groups, certain groups of Lie-type and into 26 sporadic simple groups. The precise statement of the theorem could be queries found in Section~\ref{proof-of-lemma}. Except for the 26 sporadic simple groups the other group classes are infinite. We use CFSG to prove a key lemma that allows us to handle the case for the nonabelian simple groups. 

We note that for solvable groups the design of the data structure is \emph{ independent} of CFSG. The composition factors of a solvable group are cyclic of prime order. Such cases are handled using one of the extension theorems proved in Section~\ref{extension-theorems}.

\emph{Related work}: Farzan and Munro \cite{Farzan} gave a succinct representations for finite abelian groups in a specific model of computation. In their  model \emph{a compression algorithm} first produces labels of each group element. The queries are processed by a \emph{query processing unit} which is similar to the word-RAM model. However, along with  the common arithmetic, logical and comparison  operations the query processing unit can also perform   bit-reversal in constant time. A user issuing a query, supplies the labels of two group elements that were generated by the compression algorithm to the query processing unit which then returns the label of the product of the two elements.   

Das et al. \cite{JCSS} and Das and Sharma \cite{WALCOM} have used   Erd\"os-R\'eyni cube generating sequences, Remak-Krull-Schmidt decomposition and  the structure of indecomposable groups to design their space and query-time efficient data structures. Our approach is quite different in the sense that we use the extension theorems (Section~\ref{extension-theorems}) and the Classification Theorem for Finite Simple groups to design the data structures. 

\emph{Remark}: There are several ways to represent a finite group apart from the Cayley table representation. The permutation group representation, the polycyclic presentations and the generator-relator presentations are some of the common group representations. These representations are often incomparable. For example in the generator-relator presentation we can represent infinite groups. However, many problems such as the membership testing, testing if a group is finite becomes undecidable in the generator-relator presentation (c.f. \cite{sims1}). In the permutation group representation the membership testing takes superlinear time in terms of the degree of the representation and polylogarithmic in the order of the group \cite{sims2,sims3,John-Hopcroft}. We contrast this with the Cayley table representation where membership testing can be done in constant time since the elements are known and are already used as row and column indices of the Cayley table. In the Cayley representation the user knows the labels or the names of each group element explicitly and has a direct access to each element. The labels of the elements are often taken to be $1,2,\ldots, n$ where $n$ is the order of the group. The situation is quite different for permutation group representation, polycyclic presentation or generator-relator presentation. In these cases the user does not have an explicit representation for each element.  

%-----------------------------------------------------------------------------------------------------------------------------------------------------

\section{Preliminary}\label{Preliminary}

In this section we recall some definitions and notations which we use in this paper. In this paper we only consider finite groups. The number of elements in a group $G$ is called the \emph{order of $G$} and is denoted by $|G|$. A group $G$ is abelian if $g_1 g_2=g_2 g_1$ for all $g_1, g_2 \in G.$ For a subgroup $H$ of $G$ and $g\in G$, the set $gH=\{gh \, |\, h\in H\}$ is called a left coset. Similarly, we can define right coset of $G$. The number of the left (or right) cosets of $H$ in $G$ is called the \emph{index of $H$ in $G$} and is denoted by $[G:H].$ A \emph{left traversal} of $H$ in $G$ is a set containing exactly one element from each left coset and similarly we can define right traversals. The size of left (right) traversal is the same as the index $[G:H].$ For $g \in G$, the set $gHg^{-1}=\{gag^{-1} \, | \, a \in H \}$ is called a conjugate of the subgroup $H$. A subgroup $H$ of G is said to be \emph{normal} in $G$ (denoted $H \trianglelefteq G$) if $gHg^{-1}=H$ for all $g \in G.$ We define the \emph{normalizer} of $H$ in $G$ to be the set $N_G{(H)}=\{g \in G \, | \, gHg^{-1}=H \}.$ Note that, $N_G{(H)}$ is the largest subgroup in $G$ in which $H$ is normal.

A group $G$ is called \emph{simple} if $G$ has no nontrivial normal subgroup. %the only normal subgroups of $G$ are $1_G$ and $G.$
The Classification Theorem of Finite Simple Groups states that all the finite simple groups can be classified into the following five classes: (1) cyclic groups of prime order, (2) alternating groups, (3) classical groups, (4) exceptional groups of Lie type and (5) 26 sporadic simple groups.

We list all the classes of the finite simple groups later in the Classification Theorem for Simple Groups in Section \ref{proof-of-lemma}. If $G$ is a finite simple group of Lie-type over $\mathbb{F}_q$ where $q$ is a power of some prime $p$, the Borel subgroup $B$ of $G$ is defined as the semidirect product of the Sylow $p$-subgroup of $G$ with the maximal split torus $T$. The Borel subgroup is also the normalizer of the Sylow $p$-subgroup of the finite simple group (see \cite{CA}, \cite{WLS}). 

For the purpose of this paper it might be sufficient to know some results on the \emph{orders} of certain subgroups of simple groups. The reader may choose to skip the details of the structure of these groups. We indicate what kind of subgroups we are interested in and the results regarding the order of those subgroups as and when required.
An interested reader may refer to the books by Carter \cite{CA}, Wilson \cite{WLS}, or Aschbacher \cite{ASH} for more details.

\begin{definition}[see e.g., \cite{Algebra}]
A \emph{subnormal series} of a group $G$ is chain of subgroups $$ 1= G_{k} \leq G_{k-1} \leq \dots \leq G_1   \leq G_{0}=G$$ such that $G_{i} \trianglelefteq G_{i-1},$ for all $i.$ 
\end{definition}
\begin{definition}[see e.g., \cite{Algebra}]
In a group $G$ a sequence of subgroups
$$ 1= G_{k} \leq G_{k-1} \leq \dots \leq G_1   \leq G_{0}=G$$ is called a \emph{composition series} if $G_{i} \trianglelefteq G_{i-1}$ and   $G_{i-1} / G_{i}$ is simple for all $i \in [k].$ Here, $k$ is the \emph{composition length} of $G.$
\end{definition}

\begin{theorem}[Jordan-Hölder  Theorem see e.g., \cite{Algebra}]
\noindent Let $G$ be a finite group with $G \neq 1 $. Then
\begin{enumerate}
    \item[(i)] $G$ has a composition series.
    \item[(ii)] The composition factors in a composition series are unique, namely, if $1=N_{r} \leq N_{r-1} \leq \cdots \leq N_{1} \leq N_{0}=G$ and $1=M_{s} \leq M_{s-1}\leq \cdots \leq M_{1} \leq M_{0}=G $ are two composition series for $G$, then $r=s$ and there is some permutation $\pi$ of $\{ 1,2,\dots,r\}$ such that, $$\frac{M_{\pi(i)}}{M_{\pi(i)+1}} \cong \frac{N_{i}}{N_{i+1}},  \text{for} \, 1 \leq i \leq r. $$
\end{enumerate}
\end{theorem}
\begin{theorem}[Correspondence Theorem see e.g., \cite{rot}] 
\noindent Let K $\trianglelefteq$ G and let $v: G\longrightarrow G/K$ be
the canonical map i.e. $v(g)=Kg$ for all $g.$ Then $S \mapsto v(S) = S/K$ is a bijection from the family of all those subgroups S of G which contain K to the family of all the subgroups
of $G/K$.
\end{theorem}
\vspace{0.2cm}
\textbf{Model of computation :} In this paper, we use an abstract model of computation known as the word-RAM model. In this model, data is stored in resisters and memory units. Each memory unit and resister can store $O(\log n)$ bits where $n$ is the size of the input. The unit of storage is called \emph{word}. The machine in the word-RAM model can access a word and do the usual arithmetic, logical, and comparison operations in constant time. The input size for our purpose is the order of the group. Without loss of generality, we can assume that the elements of groups are $1,2,3,...,n$. Thus, every group element can be stored in a word and can be accessed in constant time.

There are two phases in the construction of a data structure: \emph{the preprocessing phase} and \emph{the query phase}. In the preprocessing phase, we assume that we have been given a finite group by its Cayley table. Using the Cayley table, we construct a data structure that consists of some arrays and tables. In the query phase, we process multiplication queries. In a multiplication query, two group elements $g_1$ and $g_2$ are given by the user. The task is to find the product of $g_1$ and $g_2$. In this phase, the data structure constructed in the preprocessing phase is accessed to answer the query. The time taken to answer a single query is called the {\it query-time}.

The time and space used in preprocessing stage to build the data structure are generally not considered. However we show that the data structure in our case can be computed in polynomial time. Our primary concern is to consider the space used by the constructed data structure and the time it takes to answer a query to multiply the group elements.

\begin{definition}
Let $G$ be a group and $s$ and $t$ be two positive real numbers. We say that $G$ has an $(s,t)$-data structure, if $G$ can be stored in a data structure that uses at most $s$ space and can answer  a multiplication query in time at most $t$.
\end{definition}

\begin{definition}
Let $\mathcal{G}$ be a class of group and let $s,t: \mathbb{N} \rightarrow \mathbb{R}_{\geq 0}$ be two functions. If for every group $G \in \mathcal{G}$ of order $n$ there is a data structure that uses $O(s(n))$ space to store $G$ and can answer a multiplication query in time at most $O(t(n))$ then we say that $\mathcal{G}$ has an $(O(s(n)),O(t(n)))$-data structure.
\end{definition}

%-----------------------------------------------------------------------------------------------------------------------------------------------------

 \section{Extension Theorems}\label{extension-theorems}
In this section, we discuss how to use  data structures for subgroups to build new data structures for groups containing the subgroups. 
 
\begin{theorem}\label{thm-quotient}
There exist positive constants $c$ and $d$ such that for any group $G$ and a subgroup $H$ of $G$ if $H$ has an $(s,t)$-data structure for some $s$ and $t$ then $G$ has an $(s+c([G:H]^2+|G|),2t+d)$-data structure.
\end{theorem}
\begin{proof}
First we fix a left traversal $L$ and a right traversal $R$ of $H$ in $G$. Each $g\in G$ can be uniquely written as $g=hr$ where $h\in H$ and $r\in R$. Thus we can define functions $s_R:G\longrightarrow H$ and $c_R:G\longrightarrow R$ such that $g=s_R(g)c_R(g)$. Similarly we can define $c_L:G\longrightarrow L$ and $s_L:G\longrightarrow H$ such that $g=c_L(g)s_L(g)$. We can store these four functions in four arrays each of length $|G|$. 

Suppose we need to find the product of $g_1$ and $g_2$. Note that, 
\[g_1g_2=c_L(g_1)s_L(g_1)s_R(g_2)c_R(g_2).\]

Since $s_L(g_1), s_R(g_2)\in H$, we can use the data structure for $H$ to find $s_L(g_1)s_R(g_2)$ within time $t$. Let $h_1=s_L(g_1)s_R(g_2)$. Therefore, we can write  $g_1g_2=c_L(g_1)h_1c_R(g_2)$. 

Given $l\in L$ and $h\in H$, we know that there exist unique elements $h'\in H$ and $r\in R$ such that $lh=h'r$. Thus, we can define two functions  $Flip_H:L\times H\longrightarrow H$ and $Flip_R:L\times H\longrightarrow R$ such that $lh=Flip_H(l,h)Flip_R(l,h)$. We can store $Flip_H$ and $Flip_R$  in two $2$-dimensional arrays using space linear in $|H\times L|=|G|$. With the help of these functions, we can write
\[g_1g_2=Flip_H(c_L(g_1), h_1)Flip_R(c_L(g_1), h_1)c_R(g_2)=h_2r_1r_2\] where $h_2=Flip_H(c_L(g_1), h_1)$, $r_1=Flip_R(c_L(g_1), h_1)$ and $r_2=c_R(g_2)$.

Again we use the fact that any element $g$ of $G$ can be uniquely written as $g=hr$ where $h\in H$ and  $r\in R$  to define the functions $Cross_H: R\times R\longrightarrow H$ and $Cross_R: R\times R\longrightarrow R$ such that for all $r,r'\in R$ we have $rr'=Cross_H(r,r')Cross_R(r,r')$. Note that we can store these functions in two $2$-dimensional arrays each requiring size linear in $|R\times R|=(|G|/|H|)^2$.  With the help of these functions we can  write  \[g_1g_2=h_2Cross_H(r_1,r_2)Cross_R(r_1,r_2)=h_2h_3r_3\] where  $Cross_H(r_1,r_2)=h_3$ and $r_3=Cross_R(r_1,r_2)$.

Again we can use the data structure for $H$ to compute the product $h_4=h_2h_3$ within time $t$. Thus $g_1g_2=h_4r_3$. 
Finally, we define a function $Fuse:H\times R\longrightarrow G$ simply as $Fuse(h,r)=hr$ for all $h\in H$ and $r\in R$. Clearly, a $2$-dimensional array to store $Fuse$ would take space linear in $|H\times R|=|G|$. Thus, to produce the final result we just return $g_1g_2=Fuse(h_4, r_3)$.

All the functions except for $Cross_R$ and $Cross_H$ take space linear in $|G|$, while $Cross_R$ and $Cross_H$ take space linear in $(|G|/|H|)^2$. The data structure for $H$ takes space at most $s$. Therefore, the total space required is linear in $|G|+(|G|/|H|)^2$. We note that each function defined in this proof is queried exactly once. Thus, the time to query all the nine functions is bounded by some constant d. Additionally, the time taken to query the data structure for $H$ is at most $2t$. Therefore, we have the required data structure for $G$.
\end{proof}
\begin{remark}\label{remark-subgroup-ds}
We note that if the $(s,t)$-data structure for $H$ is given then the above data structure for $G$ can also be \emph{computed} in polynomial time.
\end{remark}

An immediate corollary of the above theorem is the following.
\begin{corollary}\label{cor-near-root-n}
Let $0<c_1\leq c_2$ be two constants. Let $\mathcal{G}_{c_1,c_2}$ be the class of groups $G$ that has a subgroup H with $c_1\sqrt{|G|}\leq |H|\leq c_2 \sqrt{|G|}$. Then $\mathcal{G}_{c_1,c_2}$ has $(O(n),O(1))$-data structures.
\end{corollary}
\begin{proof}
The Cayley table for $H$ takes size at most $c_2^2|G|$ and answers queries in constant time. Since $|G|/|H|\leq (1/c_1) \sqrt{|G|}$, we have $(|G|/|H|)^2\leq (1/c_1)^2 |G|$. Hence the result follows from Theorem~\ref{thm-quotient}. 
\end{proof}

In the next theorem we show how to use the data structure for a normal subgroup of a group to build a data structure for the group when the quotient group is cyclic.

\begin{theorem}\label{thm-cyclic}
There are positive constants $c$ and $d$ such that for every group $G$ and any normal subgroup $N$ of $G$, if $G/N$ is cyclic and $N$ has an $(s,t)$-data structure for some $s$ and $t$, then $G$ has an  $(s+c|G|,2t+d)$-data structure. 
\end{theorem}
\begin{proof}
Since  $G/N$ is cyclic it is generated by an element $g_0N$ where $g_0\in G$. The cosets of $N$ in $G$ are $N,g_0N,g_0^2N,\ldots,g_0^{k-1}N$ where $k$ is the order of the group $G/N$, i.e., $k=[G:N]$. Clearly, $k\leq |G|$. Let $S=\{0,1,\ldots,k-1\}$.

The set $\{g_0^0, g_0^1,\ldots,g_0^{k-1}\}$ is a left as well as a right traversal of $N$ in $G$. Hence any element $g$ could be uniquely written as $g=g_0^rn=n'g_0^r$ for some  $r \in S$ and $n,n'\in N$. This enables us to define functions $e:G\longrightarrow S$, $s_R:G\longrightarrow N$ and $s_L: G\longrightarrow N$ such that for all $g\in G$
\[g=g_0^{e(g)}s_R(g)=s_L(g)g_0^{e(g)}.\]
These three functions could be stored in arrays each having size $|G|$. 
To multiply $g_1$ and $g_2$ we first observe that $g_1g_2=g_0^{e(g_1)}s_R(g_1)s_L(g_2)g_0^{e(g_2)}$. These expression could be obtained by querying each of the functions once. The product $n_1=s_R(g_1)s_L(g_2)$ can be obtained using the data structure for $N$ within query-time $t$. Thus $g_1g_2=g_0^\alpha n_1g_0^\beta$, where $\alpha=e(g_1)$ and $\beta=e(g_2)$.

Next we define a function $Flip: N\times S\longrightarrow N$ with the property that for all $n\in N$ and $i\in S$, $ng_0^i=g_0^iFlip(n,i)$. In other words, $Flip(n,i)$ is just $g_0^{-i}ng_0^i$. This function can be stored in space linear in $|N\times S|=|G|$.
Now we can write $g_1g_2=g_0^\alpha g_0^\beta Flip(n_1,\beta)=g_0^{\alpha+\beta}n_2$ where $n_2=Flip(n_1,\beta)$. 

Next we compute $g_0^\alpha g_0^\beta=g_0^{\alpha+\beta}$. Observe that $\alpha+\beta\in \{0,1,\ldots, 2k-2\}$. We define two functions $red_e:\{0,1,\ldots, 2k-2\}\longrightarrow S$ and $red_N:\{0,1,\ldots, 2k-2\}\longrightarrow N$ such that $g_0^\ell=g_0^{red_e(\ell)}red_N(\ell)$ for all $\ell\in\{0,\ldots,2k-2\}$. Note that for $\ell<k$, $red_e(\ell)=\ell$ and $red_N(\ell)=id$. These two functions can be stored using space linear in $k$. Since $k\leq |G|$, the space required is at most linear in $G$.

Therefore, \[g_1g_2=g_0^{\alpha+\beta}n_2=g_0^{red_e(\alpha+\beta)}red_N(\alpha+\beta)n_2.\] As before the product $n_3$ of $red_N(\alpha+\beta)$ and $n_2$ can be found using the data structure for $N$. Let $red_e(\alpha+\beta)=\gamma$. Hence, $g_1g_2=g_0^\gamma n_3$.

We finally define a function $Fuse: S\times N\longrightarrow G$ as $Fuse(i,n)=g_0^in$ for all $i\in S$ and $n\in N$. Clearly, the function $Fuse$ can be stored using space linear in $|G|$. 

The product $g_1g_2$ is just $Fuse(\gamma, n_3)$.

Each function defined in this proof takes space linear in $|G|$ and the data structure for $N$ takes space at most $s$. Each function is queried exactly once and the data structure for $N$ is queried twice. This proves the theorem.   
\end{proof}
\begin{remark}
This is similar to Remark \ref{remark-subgroup-ds}. In this case too, if the data structure for $H$ in Theorem~\ref{thm-cyclic} is given then the data structure for $G$ can be computed in polynomial time.

\end{remark}

\section{Compact Data Structures for Finite Groups}\label{Compact Data structure}

Let $G$ be a group of order $n$. Our goal is to design a constant query-time data structure for $G$ of size linear in $n$.
% If G has a subgroup $H$ of order within  a constant factor of $\sqrt {|G|}$ then by Corollary~\ref{cor-near-root-n} we can find a linear sized constant query data structure for $G$. Of course there may not be such a group, e.g., a group of prime order. 
We first consider a composition series $1=G_k\lhd \ldots G_1\lhd G_0=G$ of $G$.  In case there is a subgroup $G_i$ in the composition series with size within a constant factor of $\sqrt{n}$, we can apply Corollary~\ref{cor-near-root-n} to obtain a $(O(n),O(1))$-data structure for $G$. Otherwise we consider the smallest subgroup $G_i$ of order more than $\sqrt{n}$. Note that here $|G_{i+1}|$ is at most $\sqrt{n}$ and therefore $G_{i+1}$ will have its  Cayley table of size at most $n$. This Cayley table can be used to answer a multiplication query involving elements in $G_{i+1}$ in constant time.

Now we consider the composition factor $G_i/G_{i+1}$. This quotient is a simple group. If this is an abelian group it must be cyclic (of prime order) and we can use Theorem~\ref{thm-cyclic} to get a data structure for $G_{i}$. Then an application of Theorem~\ref{thm-quotient} with $G$ and its subgroup $G_i$ will give us the required data structure for $G$. 

The nontrivial case is when $G_i/G_{i+1}$ is nonabelian. This is where we use the Classification Theorem of Finite Simple Groups. The classification theorem allows us to split the nonabelian case into various subcases. In each of the subcases we  show that we can insert two subgroups $G_{i_2}$ and $G_{i_1}$ such that $G_{i+1}<G_{i_2}<G_{i_1}<G_i$ in such a manner that the indices $[G_{i_2}:G_{i+1}]$, $[G_{i_1}:G_{i_2}]$ and $[G_i:G_{i_1}]$ are all ``small''. Since $G_{i+1}$ already has a constant query-time data structure (namely its Cayley table) of size linear in $n$,  this allows us to use Theorem~\ref{thm-quotient} successively to the group and subgroup pairs $(G_{i_2},G_{i+1})$, $(G_{i_1},G_{i_2})$, and  $(G_{i},G_{i_1})$ to obtain a constant query-time data structure for $G_i$ of size linear in $n$. Finally,  another application of Theorem~\ref{thm-quotient} with $G$ and its subgroup $G_i$ will give us the required data structure for $G$.

\subsection{Solvable Finite Groups}
In this subsection we consider the class $\mathcal{G}_{solv}$ of finite solvable groups. We  do this case first before going to the general case for the class of all finite groups because it is independent of the Classification Theorem for Finite Simple Groups.

\begin{theorem}\label{thm-solvable}
The class $\mathcal{G}_{solv}$ has $(O(n),O(1))$-data structures.
\end{theorem}

\begin{proof}
Let $G$ be a group and $1=G_k\lhd \ldots G_1\lhd G_0=G$ be a composition series of $G$. Let $n=|G|$.

\emph{Case 1}: There is $i$ such that $\sqrt{n}/2 \leq |G_i| \leq \sqrt{n}$. We simply apply Corollary~\ref{cor-near-root-n} to get the desired data structure.

\emph{Case 2:} There is no $i$ such that $\sqrt{n}/2 \leq |G_i| \leq \sqrt{n}$. Let $i$ be the largest index such that $\sqrt{n}<|G_i|$. We will have $|G_{i+1}|<\sqrt{n}/2$. The Cayley table for $G_{i+1}$ has at most $n/4$  entries. Since $G$ is solvable $G_{i}/G_{i+1}$ is cyclic of prime order. This allows us to use Theorem~\ref{thm-cyclic} to obtain a constant query-time data structure for $G_i$ which is linear in $n$. Next we observe that $[G:G_{i}]$ is at most $\sqrt{n}$. If we apply Theorem~\ref{thm-quotient} on $G$ and its subgroup $G_{i}$ we get the required data structure for $G$.  
\end{proof}

Since a composition series of a group can be computed in polynomial time, it is easy to check that the data structure in Theorem~\ref{thm-solvable} can be constructed in polynomial time.
 
 \subsection{The General Case}
 Before considering the case for general finite groups we need the following result for nonabelian simple groups.
 \begin{lemma}\label{lem-simple}
 There are positive constants $b_1$ and $b_2$ such that for any nonabelian simple group $H$ there exist subgroups $H_1$ and $H_2$ such that $1\leq H_2 \leq H_1 \leq H$ and $|H_2|\leq \sqrt{|H|}$, $[H:H_1]\leq b_1\sqrt{|H|}$, and  $[H_1:H_2]\leq b_2\sqrt{|H|}$.
 
 \end{lemma}
 \begin{proof}
 The proof uses the Classification Theorem of Finite Simple Group (CFSG). The proof idea is given in Section~\ref{proof-of-lemma} and the details are given in the Appendix.
\end{proof}
 
 Next we prove the main theorem of the paper. We note that Case 2 in the proof of the following theorem can be viewed as a generalized version of the problem of designing linear space and constant query-time data structure for nonableian simple groups. 
 
 \begin{theorem}\label{ds-for-finite-group}
 The class $\mathcal{G}_{fin}$ of all finite groups has $(O(n),O(1))$-data structures.
 \end{theorem}
 \begin{proof}
 Let $G$ be a group of order $n$. We start by  considering a composition series $1=G_k\lhd \ldots G_1\lhd G_0=G$ be a composition series of  $G$. 
 
 \emph{Case 1:} This is the case when there is $i$ such that $\sqrt{n}/2 \leq |G_i| \leq \sqrt{n}$. This case is exactly similar to the case for solvable groups.
 
 \emph{Case 2:} As before in this case we assume that there is no composition series element $G_i$ with order more that $\sqrt{n}/2$ but less than $\sqrt{n}$. Let $i$ be the largest index such $\sqrt{n}<|G_i|$. We will then have $|G_{i+1}|< \sqrt{n}/2$. Clearly, the Cayley table of $G_{i+1}$ will have at most $n/4$ entries. Since $[G:G_i]<\sqrt{n}$, by Theorem~\ref{thm-quotient} it is enough to design constant query-time data structure for $G_i$ of size linear in $n$. In the rest of the proof we therefore concentrate on designing a constant query-time data structure for $G_i$ that uses $O(n)$ space.
 
  If the composition factor $G_i/G_{i+1}$ is abelian then we are again in the same situation as in the second case of solvable groups. Therefore we assume that $G_i/G_{i+1}$ is nonabelian. %If $G_i/G_{i+1}$ is a\textbf{ sporadic} simple sporadic group or the Tits group (27th sporadic group) then the order of $G_i$ is at most constant factor bigger than the order of $G_{i+1}$. It could be a very large constant but a constant nevertheless. Therefore, the Cayley table of $G_i$ will also be a constant factor bigger than the size of the Cayley table of $G_{i+1}$ which has at most $n/4$ entries.
 
 %We now assume that $G_i/G_{i+1}$ is not a sporadic simple group or the Tits group (27th sporadic group). 
 We apply Lemma~\ref{lem-simple} to $H=G_i/G_{i+1}$ to obtain subgroups $H_1$ and $H_2$ such that $1\leq H_2 \leq H_1 \leq H=G_i/G_{i+1}$. By the correspondence theorem of groups, $H_1$ and $H_2$ will be of the form $G_{i_1}/G_{i+1}$ and $G_{i_2}/G_{i+1}$ respectively for some subgroups $G_{i_1}$ and $G_{i_2}$ such that $G_{i+1} \leq G_{i_2}\leq G_{i_1}\leq G_i$.
 From Lemma~\ref{lem-simple} we have $[H_1:H_2]\leq b_2 |H|$. Since, $H_1=G_{i_1}/G_{i+1}$ and $H_2=G_{i_2}/G_{i+1}$, we have $[G_{i_1}/G_{i+1}: G_{i_2}/G_{i+1}]\leq b_2\sqrt{|G_i/G_{i+1}|}\leq b_2\sqrt{n}$.
 
 Therefore, $[G_{i_1}:G_{i_2}]\leq b_2 \sqrt{n}$. Similarly, $[G_{i}:G_{i_1}]\leq b_1\sqrt{n}$. Again from Lemma~\ref{lem-simple}, we have $H_2\leq \sqrt{|H|}$. This implies, $[G_{i_2}:G_{i+1}]\leq \sqrt{|G_i/G_{i+1}|}\leq \sqrt{n}$.
 
 Since $G_{i+1}$ has a Cayley table of size at most $n$ and $[G_{i_2}:G_{i+1}]\leq \sqrt{n}$, we will have a constant query-time data structure for the subgroup $G_{i_2}$ of size at most $n$ by Theorem~\ref{thm-quotient}. Since $[G_{i_1}:G_{i_2}]\leq b_2 \sqrt{n}$ and $[G_{i}:G_{i_1}]\leq b_1\sqrt{n}$, another two applications of Theorem~\ref{thm-quotient} with the group and subgroup pairs $(G_{i_1},G_{i_2})$ and $(G_{i},G_{i_1})$ will give a data structure for $G_i$ of size linear in $n$ which can answer a multiplication query in constant time.   
 \end{proof}
  
%We note that there exist polynomial time algorithms for finding a composition series \cite{Ako-Seress} and checking if a composition factor is abelian \cite{Kavitha-T}. It seems plausible to have polynomial time algorithm  to compute $G_{i_1}$ and $G_{i_2}$ in Case 2 of the proof of the above theorem. However we have not checked this for all the cases of CFSG because obtaining polynomial prepossessing time is not our main focus. Nevertheless, we note that $G_{i_1}$ and $G_{i_2}$  can be found simply by a brute force approach. Therefore, we can actually {\it construct} the data structure for $G$ in the above theorem.

A discussion on how to construct the above data structure in polynomial time is given at the end of Section \ref{proof-of-lemma}.

\section{Proof Sketch for Lemma~\ref{lem-simple}}\label{proof-of-lemma}
  In this section we sketch the proof idea behind Lemma~\ref{lem-simple}. We first state the Classification Theorem of Finite Simple Groups.
 
\begin{theorem}[\cite{WLS}]{(The Classification Theorem of Finite Simple Group)}\label{classification-thm}

Every finite simple group is isomorphic to one of the following:
\begin{enumerate}
    \item[(i)] a cyclic group $C_p$ of prime order $p$;
    \item[(ii)] an alternating group $A_m,$ \,for $m \geq 5$;
    \item[(iii)]  a classical group;
    \begin{enumerate}
        \item linear:\,\,\quad \,\,\, $ A_{\rm m}(q) ( \text{or}\,\, { \rm PSL}_{{\rm m}+1}(q) ), {\rm m} \geq 1$, except ${\rm PSL}_2(2)$ and ${\rm PSL}_2(3)$;
        \item unitary: \quad \, $^2A_{\rm m}(q^2) ( \text{or}\, {\rm PSU}_{{\rm m}+1}(q))\,  , {\rm m} \geq 2,$  except ${ \rm PSU}_{3}(2)$;
        \item symplectic:  $C_{\rm m}(q)) ( \text{or~}\, {\rm PS}_{{\rm p}_{2{\rm m}}}(q)) , {\rm m} \geq 2$  except ${{\rm PS}_{p}}_{4}(2)$;
        \item orthogonal: $B_{\rm m}(q)  ( \text{or~}\,{\rm  P}\Omega_{2{\rm m}+1}(q)) , {\rm m} \geq 3, q$ odd;

        \hspace{2cm}$D_{\rm m}(q) ( \text{or~}\, {\rm P} \Omega_{2{\rm m}}^{+}(q))  , {\rm m} \geq 4$;
        
        \hspace{2cm}$^2D_{\rm m}(q^2)  ( \text{or~}\, {\rm P}\Omega_{2{\rm m}}^{-}(q))\,, {\rm m} \geq 4$      
    \end{enumerate}
    where $q$ is a power $p^a$ of some prime;
    \item[(iv)] an exceptional group of Lie type:
    $$\rm{G_2(q), q \geq3; F_4(q); E_6(q); ^2E_6(q); ^3D_4(q); E_7(q); E_8(q) \,\, \text{or} }$$ 
    where $q$ is a power $p^a$ of some prime;
    $$\rm{^2B_2(2^{2m+1})}, {\rm m} \geq 1; ^2G_2(3^{2{\rm m}+1}), {\rm m} \geq 1; ^2F_4(2^{2{\rm m}+1}), {\rm m} \geq 1 $$ 
    or the Tits group $^2F_4(2)^{'}$;
    \item[(v)] one of 26 sporadic simple groups:
    \begin{enumerate}
        \item the five Mathieu groups $\rm {M_{11}, M_{12}, M_{22}, M_{23}, M_{24}}$;
        \item the seven Leech Lattice groups $\rm{Co_1, Co_2, Co_3, McL, HS, Suz, J_2}$;
        \item the three Fischer groups $\rm {Fi_{22}, Fi_{23}, Fi_{24}^{'};}$
        \item the Monstrous groups $\mathbb{M}, \mathbb{B},\rm{ Th, HN, He}$;
        \item the six pariahs $\rm{J_1, J_2, J_4, O'N, Ly, Ru.}$
    \end{enumerate}
\end{enumerate}
\end{theorem}
The definition of each of the group classes mentioned in the above theorem can be found in the standard texts on CFSG (see e.g., \cite{CA}, \cite{WLS}, \cite{ASH}).  

Since Lemma~\ref{lem-simple} is about nonabelian simple groups we need to consider cases (ii) to (v)  in Theorem~\ref{classification-thm}. We take each subcases under these cases and show that there are subgroups $H_1$ and $H_2$ satisfying the conditions of the lemma.

 We note that the 26 sporadic simple groups listed in the case (v) are of constant sizes. Therefore, we can ignore these groups for the purpose of the proof by simply taking $H_2$ to be the identity subgroup and $H_1$ to be $H$. Of course if we do so we need to pick extremely large constant $b_2$ as some the sporadic simple groups are of huge sizes. Fortunately, there are known results on the groups listed under case (v) that helps us to keep the constants  $b_1$ and $b_2$ under $5$.

 For an  alternating group $A_m$ (case (ii)) it is easy to show that $H_1=A_i$ and $H_2=A_{i-1}$ for some suitably chosen $i<m$ does the job. The details are in the Appendix.

For the remaining groups we use the following two methods for the choices of $H_1$ and $H_2$. The methods are as follows:

\begin{enumerate}
    \item \textbf{Method 1:} In this method, we first choose $H_2$ to be a certain Sylow subgroup of the given simple group $H$. Next we pick $H_1$ to be the normalizer of $H_2$ in $H$ or the Borel subgroup containing $H_2$.  
    
    ~
    
    \emph{Example:} Let us take $H$ to be a simple group $A_m(q)$ for some $q>2$ which appears in case (iii) of Theorem~\ref{classification-thm}. Here $q$ is power of some prime $p$. It is known that $A_m(q)$ has order $q^{{m(m+1)}/{2} }  \prod_{i=1}^{m} (q^{i+1}-1)/{(q-1,m+1)}$ where $(q-1,m+1)$ denotes the gcd of $q-1$ and $m+1$ (see \cite{ASH}, p. 252). Clearly, $H$ will have a Sylow  $p$-subgroup of order $q^{{m(m+1)}/{2}}$. We set  $H_2$ to be this subgroup. Next we pick $H_1$ to be the normailzer of $H_2$ in $H$. It is also known that the order of $H_1$ is $q^{{m(m+1)}/{2} }(q-1)^m$ (see \cite{WLS}, p. 46). One can check that with $b_1=2$ and $b_2=1$, these choices satisfy the conditions of Lemma~\ref{lem-simple} (see Appendix for the details).
    
    ~
    
    \item \textbf{Method 2:} In this method, we choose $H_1$ to be a maximal subgroup of the simple group $H$ and $H_2$ to certain Sylow subgroup of $H_1$.
    
    ~
    
    \emph{Example:} In the example under Method 1 we consider the case for $A_m(q)$ when $q>2$. In this example we take the case when $q=2$. Here  $H=A_m(q)$ will have order  $2^{m(m+1)/2}\prod_{i=1}^{m} (2^{i+1}-1)$ (see \cite{ASH}, p. 252). It is known that the maximal subgroup of $H$ is of order $|H|/(2^m-1)$ (see \cite{KL}, p.  175). We take this subgroup as $H_1$. Next we take $H_2$ as a Sylow $2$-subgroup of $H_1$ which has order $2^{m(m+1)/2}$. It is easy to verify that these choices of $H_1$ and $H_2$ along with $b_1=b_2=1$ satisfy the conditions of Lemma~\ref{lem-simple} (see Appendix for the details).
     
\end{enumerate}

Table \ref{Table_1} lists the methods that we have used for choosing the suitable subgroups in the corresponding nonabelian simple group. The last two columns represent the constant factors $b_1$ and $b_2$ for the corresponding simple group (see Table \ref{Table_1}). 

For case (v), we use Method 2 to get the suitable subgroups.

%For case (v), we use Method 2 to get the suitable subgroups.
\begin{table}
\centering
\begin{tabular}{|l|l|l|l|l|l|}
\hline
Case & $H$ & Condition on $q$ & Method  & $ b_1 $ & $ b_2 $ \\ \hline
\multirow{13}{*}{(iii)} & \multirow{2}{*}{$A_{m}(q)$} & $q>2$ & Method 1 & $2$ & $1$  \\ \cline{3-6} 
 &  & $q=2$ & Method 2 & $1$  & $1$   \\ \cline{2-6} 
 & \multirow{3}{*}{$ ^2A_m(q^2); m >1$} & $q>2$ & Method 1  & $2$ & $1$  \\ \cline{3-6} 
 & & $q=2; 6 \nmid (m-1)$  & Method 2 & $1$& $ 1 $\\ \cline{3-6} 
 &  & $q=2; 6 \mid (m-1) $  & Method 2 & $1$ &  $1$ \\ \cline{2-6}
 & \multirow{2}{*}{$C_{m}(q); m>2$} & $q>2$ & Method 1 & $2$  & $1$  \\ \cline{3-6} 
 &  & $q=2$ & Method 2 & $1$ & $1$    \\ \cline{2-6} 
 & $B_{m}(q); m>1$ & $q$ odd & Method 1 & $2$ & $1$  \\ \cline{2-6} 
 & \multirow{2}{*}{$D_{m}(q); m>3$} & $q>2$ & Method 1 & $2$  &  $1$ \\ \cline{3-6} 
 &  & $q=2$ & Method 2 & $1$ &    $1$\\ \cline{2-6}
 & \multirow{2}{*}{$^2D_m(q^2); m>3$} & $q>2$ & Method 1 & $3$ & $1$  \\ \cline{3-6} 
 &  & $q=2$ & Method 2 & $1$ &  $1$  \\ \cline{1-6} 
 
 \multirow{14}{*}{(iv)} & $G_2(q)$ & $q \geq 3$ & Method 1 & $1$ & $1$  \\ \cline{2-6} 
 & $F_{4}(q)$  & All $q$ & Method 2 & $1$ & $1$   \\ \cline{2-6}
 &\multirow{2}{*}{ $E_{6}(q)$}  & $q>2$ & Method 1 & $1$ & $1$   \\ \cline{3-6}
 &  & $q=2$ & Method 2 & $ 1 $ & $ 1 $   \\ \cline{2-6}
 
 & $^2E_6(q)$  & All $q$ & Method 1 & $1$ & $1$   \\ \cline{2-6}
 & $^3D_4(q)$  & All $q$ & Method 1 & $1$ & $1$   \\ \cline{2-6}
 
 &\multirow{2}{*}{$E_{7}(q)$} & $q>2$ & Method 1 & $1$ & $1$   \\ \cline{3-6}
 &   & $q=2$ & Method 2 & $ 1 $ & $ 1 $   \\ \cline{2-6}
 
 &\multirow{2}{*}{$E_{8}(q)$}  & $q>2$ & Method 1 & $1$ & $1$   \\ \cline{3-6}
 &   & $q=2$ & Method 2 & $ 1 $ & $ 1 $   \\ \cline{2-6}
 
 & $^2B_2(q)$  & $q=2^{2t+1},t \geq 1$ & Method 1 & $1$ & $1$   \\ \cline{2-6}
 & $^2G_2(q)$   & $q=3^{2t+1},t \geq 1$ & Method 1 & $1$ & $1$    \\ \cline{2-6}
 & $^2F_4(q)$  & $q=2^{2t+1}, t \geq 1 $ & Method 1 & $1$  & $1$   \\ \cline{2-6}
 & $^2F_4(2)'$  & $q=2$ & Method 2 &  $1$  & $ 1 $   \\ \cline{1-6}
\end{tabular}

\caption{Table representing the constant factor and method used for choosing suitable subgroups}
\label{Table_1}
\end{table}

%For details of all the cases and subcases, please see Appendix~\ref{AP} ***. 
In Appendix, Section ~\ref{tables} contains two comprehensive tables listing the orders of subgroups used in the proof of Lemma \ref{lem-simple} for different cases of CFSG.

\begin{remark}
Now we discuss how to construct the data structure given in Theorem \ref{ds-for-finite-group} in polynomial time.
Apart from some simple polynomial time computations, one can check that to construct the data structure in polynomial time it is enough to perform the following four tasks in polynomial time.

\begin{enumerate}
\item Computing a composition series,
\item Computing a Sylow $p$-subgroup,
\item Computing the normalizer of a subgroup, and 
\item Computing maximal subgroups of simple groups. 
\end{enumerate}

It is well known that the tasks (1), (2) and (3) can be performed in polynomial time even for permutation groups (See e.g., \cite{Ako-Seress}). For (4), we note that any maximal subgroup of a simple group is generated by at most 4 elements \cite{maximal-simple}. Therefore, we can in fact enumerate all the maximal subgroups of a simple group given by its Cayley table in polynomial time.

There is one more subtlety: How do we determine which method (Method 1 or Method 2) to apply to get the correct subgroups as required by Lemma 11 and Theorem 12? There are two ways to address this issue. a) We just apply both the methods. One of them is bound to produce the subgroups of required sizes by Lemma 11. b) We can identify the type of the simple group by doing isomorphism tests. Note that the isomorphism of simple groups can be tested in polynomial time as any simple group is generated by just two elements. We also need the well-known fact that there are at most 2 simple groups of any given order.

\end{remark}

\bibliographystyle{plainurl}

\bibliography{LinearDataStructure}

\begin{thebibliography}{10}

\bibitem{AI}
N.~Ahanjideh and A.~Iranmanesh.
\newblock On the {S}ylow normalizers of some simple classical groups.
\newblock {\em Bull. Malays. Math. Sci. Soc. (2)}, 35(2):459--467, 2012.

\bibitem{Arvind-Jacobo}
Vikraman Arvind and Jacobo Tor{\'a}n.
\newblock The complexity of quasigroup isomorphism and the minimum generating
  set problem.
\newblock In {\em International Symposium on Algorithms and Computation}, pages
  233--242. Springer, 2006.

\bibitem{Jacobo-Arvind}
Vikraman Arvind and Jacobo Tor{\'{a}}n.
\newblock Solvable group isomorphism is (almost) in {NP} {\(\cap\)} conp.
\newblock {\em {ACM} Trans. Comput. Theory}, 2(2):4:1--4:22, 2011.

\bibitem{ASH}
M.~Aschbacher.
\newblock {\em Finite Group Theory}.
\newblock Cambridge Studies in Advanced Mathematics. Cambridge University
  Press, 2 edition, 2000.

\bibitem{Babai}
L{\'a}szl{\'o} Babai, Paolo Codenotti, and Youming Qiao.
\newblock Polynomial-time isomorphism test for groups with no abelian normal
  subgroups.
\newblock In {\em International Colloquium on Automata, Languages, and
  Programming}, pages 51--62. Springer, 2012.

\bibitem{Babai-1}
L{\'a}szl{\'o} Babai and Youming Qiao.
\newblock Polynomial-time isomorphism test for groups with abelian sylow
  towers.
\newblock In {\em STACS'12 (29th Symposium on Theoretical Aspects of Computer
  Science)}, volume~14, pages 453--464. LIPIcs, 2012.

\bibitem{maximal-simple}
Timothy~C Burness, Martin~W Liebeck, and Aner Shalev.
\newblock Generation and random generation: from simple groups to maximal
  subgroups.
\newblock {\em Advances in Mathematics}, 248:59--95, 2013.

\bibitem{CA}
Roger~W. Carter.
\newblock {\em Finite groups of {L}ie type}.
\newblock Wiley Classics Library. John Wiley \& Sons, Ltd., Chichester, 1993.
\newblock Conjugacy classes and complex characters, Reprint of the 1985
  original, A Wiley-Interscience Publication.

\bibitem{ATLAS}
J.~H. Conway, R.~T. Curtis, S.~P. Norton, R.~A. Parker, and R.~A. Wilson.
\newblock {\em {ATLAS} of Finite Groups}.
\newblock Oxford University Press, Eynsham, 1985.
\newblock Maximal subgroups and ordinary characters for simple groups, With
  computational assistance from J. G. Thackray.

\bibitem{WALCOM}
Bireswar Das and Shivdutt Sharma.
\newblock Compact data structures for dedekind groups and finite rings.
\newblock In {\em WALCOM}, pages 90--102, 2021.

\bibitem{JCSS}
Bireswar Das, Shivdutt Sharma, and P.~R. Vaidyanathan.
\newblock Space efficient representations of finite groups.
\newblock {\em J. Comput. Syst. Sci.}, 114:137--146, 2020.

\bibitem{Algebra}
David~S. Dummit and Richard~M. Foote.
\newblock {\em Abstract algebra}.
\newblock John Wiley \& Sons, Inc., Hoboken, NJ, third edition, 2004.

\bibitem{Farzan}
Arash Farzan and J.~Ian Munro.
\newblock Succinct representation of finite abelian groups.
\newblock In {\em I{SSAC} 2006}, pages 87--92. ACM, New York, 2006.

\bibitem{John-Hopcroft}
Merrick Furst, John Hopcroft, and Eugene Luks.
\newblock Polynomial-time algorithms for permutation groups.
\newblock In {\em 21st Annual Symposium on Foundations of Computer Science
  (sfcs 1980)}, pages 36--41. IEEE, 1980.

\bibitem{Le-Gall}
Francois~Le Gall.
\newblock Efficient isomorphism testing for a class of group extensions.
\newblock {\em arXiv preprint arXiv:0812.2298}, 2008.

\bibitem{Kavitha-T}
T.~Kavitha.
\newblock Linear time algorithms for abelian group isomorphism and related
  problems.
\newblock {\em J. Comput. System Sci.}, 73(6):986--996, 2007.

\bibitem{Kayal-Neeraj}
Neeraj Kayal and Timur Nezhmetdinov.
\newblock Factoring groups efficiently.
\newblock In {\em International colloquium on automata, languages, and
  programming}, pages 585--596. Springer, 2009.

\bibitem{KL}
Peter~B. Kleidman and Martin~W. Liebeck.
\newblock {\em The Subgroup Structure of the Finite Classical Groups}.
\newblock London Mathematical Society Lecture Note Series. Cambridge University
  Press, 1990.

\bibitem{Daniel-Rothschild}
Daniel~J. Kleitman, Bruce~R. Rothschild, and Joel~H. Spencer.
\newblock The number of semigroups of order {$n$}.
\newblock {\em Proc. Amer. Math. Soc.}, 55(1):227--232, 1976.

\bibitem{Kumar-Ravi}
S~Ravi Kumar and Ronitt Rubinfeld.
\newblock Property testing of abelian group operations, 1998.

\bibitem{Miller}
Gary~L. Miller.
\newblock On the {$n^{\log n}$} isomorphism technique: {A} preliminary report.
\newblock In Richard~J. Lipton, Walter~A. Burkhard, Walter~J. Savitch, Emily~P.
  Friedman, and Alfred~V. Aho, editors, {\em Proceedings of the 10th Annual
  {ACM} Symposium on Theory of Computing, May 1-3, 1978, San Diego, California,
  {USA}}, pages 51--58. {ACM}, 1978.

\bibitem{Jayalal-Sarma}
Youming Qiao, Jayalal Sarma, and Bangsheng Tang.
\newblock On isomorphism testing of groups with normal hall subgroups.
\newblock {\em J. Comput. Sci. Technol.}, 27(4):687--701, 2012.

\bibitem{rot}
Joseph~J. Rotman.
\newblock {\em An introduction to the theory of groups}, volume 148 of {\em
  Graduate Texts in Mathematics}.
\newblock Springer-Verlag, New York, fourth edition, 1995.

\bibitem{Ako-Seress}
\'{A}kos Seress.
\newblock {\em Permutation group algorithms}, volume 152 of {\em Cambridge
  Tracts in Mathematics}.
\newblock Cambridge University Press, Cambridge, 2003.

\bibitem{sims3}
Charles~C. Sims.
\newblock Computational methods in the study of permutation groups.
\newblock In John Leech, editor, {\em Computational Problems in Abstract
  Algebra}, pages 169--183. Pergamon, 1970.

\bibitem{sims2}
Charles~C Sims.
\newblock Computation with permutation groups.
\newblock In {\em Proceedings of the second ACM symposium on Symbolic and
  algebraic manipulation}, pages 23--28, 1971.

\bibitem{sims1}
Charles~C. Sims.
\newblock {\em Computation with finitely presented groups}, volume~48 of {\em
  Encyclopedia of Mathematics and its Applications}.
\newblock Cambridge University Press, Cambridge, 1994.

\bibitem{Lint-Wilson}
J.~H. van Lint and R.~M. Wilson.
\newblock {\em A course in combinatorics}.
\newblock Cambridge University Press, Cambridge, 1992.

\bibitem{WLS}
Robert~A. Wilson.
\newblock {\em The finite simple groups}, volume 251 of {\em Graduate Texts in
  Mathematics}.
\newblock Springer-Verlag London, Ltd., London, 2009.

\bibitem{max}
Robert~A. Wilson.
\newblock Maximal subgroups of sporadic groups.
\newblock In {\em Finite simple groups: thirty years of the atlas and beyond},
  volume 694 of {\em Contemp. Math.}, pages 57--72. Amer. Math. Soc.,
  Providence, RI, 2017.

\end{thebibliography}

\newpage

\newpage
\begin{center}
    \textbf{\Large{Appendix}}\label{AP}
\end{center}

\section{Proof of Lemma~\ref{lem-simple}}
In this section we prove Lemma \ref{lem-simple} in detail. We present the detailed calculations in the ordering mentioned in the Classification Theorem of Finite Simple Group, i.e., Theorem \ref{classification-thm}. As we mentioned in Section \ref{Preliminary}, we just need to use some known results on the \emph{order} of certain subgroups of simple groups. The detailed description of these groups may be skipped for the purpose of the proof. The results that are used in the proof are on the orders of the finite simple groups, on the orders of maximal subgroups of simple groups and the normalizers of certain types of Sylow subgroups of simple groups. The information about the order of these simple groups can be obtained from \cite{ASH} (see p. 252).

In case (ii) of Theorem \ref{classification-thm}, $H$ is an alternating group. The stabilizer of any set  is a subgroup of an alternating group under natural action. The subgroups $H_1$ and $H_2$ are suitably picked stabilizer subgroups of the given  alternating group $H$.  

For the cases (iii) and (iv) of Theorem \ref{classification-thm},  we use Method 1 and Method 2 to get the desired subgroups as required in Lemma \ref{lem-simple}. In this cases the finite simple group $H$  is of Lie-type and is defined over  a  finite field $\mathbb{F}_q$ where $q$ is a power of some prime $p$. In Method 1, we take $H_2$ to be certain Sylow $p$-subgroup of $H$. The existence of such $H_2$ follows from the well-known Sylow theorem. For the existence of $H_1$, we take the normalizer of $H_2$ or the Borel subgroup. The information about the order of normalizer has been obtained from (see \cite{CA}, p. 76, \cite{WLS}, p. 46).

%{\color{blue}{Since Sylow $p$-subgroup is a unipotent radical of the Borel subgroup. In general the Borel subgroup is contained in the normalizer of a unipotent subgroup. But the subgroups of G which contain the Borel subgroup are the parabolic subgroups (see, e.g., \cite{CA}, p. 43, Prop. 2.1.6 (i)). Thus, the normalizer has to be same as the Borel subgroup.}} 
 
 For the groups in which we use Method 2, we consider a maximal subgroup of $H$ as $H_1$ and $H_2$ to be some Sylow $p$-subgroup of $H_1$. The index of a maximal subgroup (and hence its order) can be obtained from \cite{KL}, p. 175 and \cite{WLS}, p. 156.
 %There is a change in notation in the reference for different groups, which can be taken care of by resembling it from (\cite{KL}, p. 170).
 
 For the simple groups in case (v), we use Method 2 and the information about order of maximal subgroup ($H_1$) can be obtained from \cite{max}. Also, for the choice of $H_2$, we choose certain Sylow subgroup of $H_1$.
 
 \vspace{0.3cm}
\noindent \textbf{Remarks}
 
The inequalities in the following two remarks are used in the calculation multiple times. For the sake of completeness we provide proofs of the inequalities.

\begin{remark}\label{Remark-1}
For all integer $q>2,$ we have  $\frac{q}{(q-1)^{2}}< 1.$
\end{remark}
\begin{proof}
\begin{align*}
    1-\frac{q}{(q-1)^{2}}=&\frac{(q-1)^2-q}{(q-1)^{2}}\\=&\frac{q^2-3q+1}{(q-1)^{2}}\\=&\frac{(q-\frac{3+\sqrt{5}}{2})(q-\frac{3-\sqrt{5}}{2})}{(q-1)^{2}} \\>& 0.
\end{align*}
Hence, $\frac{q}{(q-1)^{2}}< 1$, for $q >2$ 
\end{proof}

\begin{remark}\label{Remark-2}
$$\prod_{i=1}^{i=m}(q^{i+1}-(-1)^{i+1})<q^{\sum_{i=1}^{i=m}(i+1)}$$.
\end{remark}
\begin{proof}
 We can observe that the sign of 1 changes alternatively. When $i$ is even then sign of 1 is negative and when $i$ is odd then sign of 1 is positive.
 
 Now, $$(q^{2j-1}-(-1)^{2j-1})(q^{2j}-(-1)^{2j})=(q^{4j-1}-1-q^{2j}+q^{2j-1}) < q^{4j-1} \,( j \geq 1).$$
 
 If $m$ is even, then it is easy to see that we can pair two consecutive odd and even term, the product of these terms is less than the sum of powers of $q$. If $m$ is odd, then the term which is not paired is $q^m-1$. Notice that, $(q^m-1)< q^m$. Thus, in this case also the product of all the terms is less than the sum of powers of $q$.
\end{proof}
%\begin{remark}
%Let $K_2$ be a subgroup of a group $K_1$ and $\alpha$ be a positive real number, then $\big( \frac{|K_1|}{|K_2|})^2\leq \alpha $ if and only if $\big( \frac{|K_1|}{|K_2|})\leq\sqrt{\alpha }$.
%\end{remark}
%\begin{proof}
 %We know that $[K_1: K_2]=\frac{|K_1|}{|K_2|}>0$. Since, the square function is monotonically increasing on non-negative domain, $\big( \frac{|K_1|}{|K_2|})^2\leq \alpha $ if and only if $\big( \frac{|K_1|}{|K_2|})\leq\sqrt{\alpha }$.
%\end{proof}
\begin{remark}
 The greatest common divisor (gcd) of two natural  numbers $m$ and $n$, both not zero, is denoted by $(m,n)$.
\end{remark}
 
\subsection{Alternating group}

The \emph{Alternating group} $A_m$ is a group of all even permutations of a finite set. It is well known that $A_m$ is simple group, when $m \geq 5.$ Notice that, with respect to the natural action, the \emph{stabilizer} of a point $\alpha$ or a set $S \subset [m] $ is a subgroup of $A_m.$ In particular, the stabilizer of the set $\{ i+1,...,m \}$ is a subgroup of $ A_{m}$ and is isomorphic to $A_i.$ 
      
    Thus, there exists a subgroup $K_j$ of $A_m$ such that $K_j \cong A_j,$ $\forall$ j and $\lvert  K_j \rvert = \frac{j!}{2}.$ Let  $k \in \mathbb{Z}$ such that $\frac{k!}{2} \leq \sqrt{\frac{m!}{2}} < \frac{(k+1)!}{2},$ and  $$H_2 \cong A_k \text{ amd } H_1 \cong A_{k+1}.$$
Thus, $A_k$ is subgroup of $A_{k+1}$ as $A_k$ is stabilizer which fixes the point $ k+1$ and  $H_2 \leq H_1.$
Notice that,
    $$\lvert H_2 \rvert ^2 = \Big(\frac{k!}{2}\Big)^2 \leq \frac{m!}{2}.$$
Now, consider
    \begin{align*}
        {\frac{m!}{2}} &< \Big(\frac{(k+1)!}{2}\Big)^2\\
        \frac{(k+2)  (k+3)  \cdots m}{2} &< \frac{1\cdot 2  \cdots k+1}{4}\\
        (k+2)  (k+3)  \cdots m &< 3 \cdot 4  \cdots k+1.
    \end{align*}
 Notice that, each term in right hand side is less than that of each term in left hand side. Thus, the number of terms in left hand side must be strictly less than the number of terms in the right hand side which implies that 
 \begin{align*}
     m-(k+1) &< (k+1)-2\\
     m-2k &< 0\\
    \frac{m}{2} &< k.
 \end{align*} 
 Thus, $k >  \frac{m}{2}$ and it is easy to find such $k.$ \\
 Consider,
$$\Big(\frac{\lvert H_1 \rvert}{\lvert H_2 \rvert}\Big)^2 = \Big(\frac{\frac{(k+1)!}{2}}{\frac{k!}{2}}\Big)^2 = (k+1)^2.$$
Now, when $m=5,k=3$ and $m=6,k=4$ , we can see that the following inequality holds:
$$(k+1)^2< \frac{m!}{2}.$$
Consider, $m\geq 7$ and $k\geq 4,$ then,
\begin{align*}
    \frac{k!}{2} & \leq \sqrt{\frac{m!}{2}}\\
    \Big(\frac{k!}{2}\Big)^2 &\leq \frac{m!}{2}\\
    \frac{k^2 (k-1)^2 (k-2)^2 }{4} &\leq \Big(\frac{k!}{2}\Big)^2\leq \frac{m!}{2}.
\end{align*}
Let,
\allowdisplaybreaks
\begin{align*}
   \frac{4 (k+1)^2}{ k^2 (k-1)^2 (k-2)^2 } &= \frac{4 (1+\frac{1}{k})^2}{(k-1)^2 (k-2)^2 }\\
                                           &< \frac{4} {4 \cdot 9} \frac{25}{16}\\
                                           &<1\\
    (k+1)^2 < \frac{ k^2 (k-1)^2 (k-2)^2 }{ 4 } &\leq \Big(\frac{k!}{2}\Big)^2\leq \frac{m!}{2}.
\end{align*}
Thus, it implies that,
$$\Big(\frac{\lvert H_1 \rvert}{\lvert H_2 \rvert}\Big)^2 = (k+1)^2 \leq \frac{m!}{2}.$$
Clearly,
$$\Big(\frac{\lvert H \rvert}{\lvert H_1 \rvert}\Big)^2 = \Big(\frac{\frac{(m)!}{2}}{\frac{(k+1)!}{2}}\Big)^2 < \Big(\frac{\frac{m!}{2}}{\sqrt{\frac{m!}{2}}}\Big)^2 = \frac{m!}{2}.$$

\subsection{The Classical Groups}
In this section, we consider $H$ to be a classical simple group described in case (iii) of Theorem \ref{classification-thm}. Let $q$ be a power of some prime $p$. As described earlier, we use Method 1 and Method 2 to show the existence of subgroups $H_2$ and $H_1$ of the simple group $H$.
%as described in Method 1, we will consider $H_2$ to be some Sylow subgroup of $H$ and $H_1$ is the normalizer of $H_2.$ It is known that the normalizer of \textbf{Sylow $p$-subgroup of $H$ is} the Borel subgroup($B$) in $H$ (\cite{}, \cite{}). Thus, we will denote the Borel subgroup by $H_1,$ which is actually a normalizer of the subgroup $H_2$

\subsubsection{A Classical Groups of Linear Type:}
\begin{enumerate}
    \item[ 1.1] $H=A_{m}(q)$; $m \geq 1$, $q>2$ (Method 1)
        
        \vspace{0.1cm}
        The finite simple group $A_{m}(q)$ is isomorphic to the \emph{projective special linear group} $ \rm{PSL}_{{\rm m}+1}(q)$, where ${\rm PSL}_{m+1}(q)$ is the group obtained by taking special linear group ${\rm SL}_{m+1}(q)$ and quotienting out by its center, i.e. $A_{m}(q) \cong \frac{{\rm SL}_{m+1}(q)}{Z({\rm SL}_{m+1}(q))}$ (see \cite{WLS}, p. 44). It is known that (see \cite{ASH}, p.  252) its order is,
        \begin{center}
        $\lvert H \rvert= \frac{q^{\frac{m(m+1)}{2} }  \prod_{i=1}^{m} (q^{i+1}-1)  }{(q-1,m+1)}.$  
        \end{center}
        Let $H_2$ be the Sylow $p$-subgroup of $A_{m}(q)$, then $\lvert H_2 \rvert= q^{\frac{m(m+1)}{2}}.$ \,Notice that $ \lvert H_2 \rvert^2  \leq \lvert H \rvert.$ Consider the Borel subgroup $H_1$ of $H$, its order is (see \cite{WLS}, p. 46),
        %which the semidirect product of $H_2$ and diagonal subgroup $T$ of ${\rm PSL}_{m+1}(q)$
         $$\vert H_1 \vert=\frac{q^{\frac{m(m+1)}{2}}}{(q-1,m+1)} (q-1)^m.$$
         Consider,
\allowdisplaybreaks         
        \begin{align*}
            \frac{\lvert  H_1 \rvert}{\lvert H_2 \rvert^2} 
            &= \frac{ \frac{q^{\frac{m(m+1)}{2}}}{(q-1,m+1)} (q-1)^m }{q^{m(m+1)} }\\
            & = \frac{1}{(q-1,m+1)} \frac{(q-1)^m }{q^{\frac{m(m+1)}{2}} }\\
            & < \frac{1}{(q-1,m+1)} \frac{q^m }{q^{\frac{m(m+1)}{2}}}\\
            & < \frac{q^m }{q^{1+2+3+ \ldots +m}}\\
            & <1.
        \end{align*}
Thus,$$\Big(\frac{\vert H_1 \vert}{\lvert H_2 \rvert}\Big)^2  < \vert H_1 \vert  < \vert H\vert.$$
Also, 
        \begin{align*}
         \frac{\lvert H \rvert}{\vert H_1 \vert^2} 
            &=\frac{\frac{q^{\frac{m(m+1)}{2} }  \prod_{i=1}^{m} (q^{i+1}-1)  }{(q-1,m+1)}}{( \frac{q^{\frac{m(m+1)}{2}}}{(q-1,m+1)} (q-1)^m)^2 }\\
            &= \frac{ (q-1,m+1) \prod_{i=1}^{m} (q^{i+1}-1)  }{ q^{\frac{m(m+1)}{2}} (q-1)^{2m}}\\
            & < \frac{ q^{1+2+ \ldots + m} q^{m+1}}{q^{\frac{m(m+1)}{2}}(q-1)^{2m} }\\
            & = \frac{  q^{m+1}}{(q-1)^{m+1} (q-1)^{m-1} }\\
            &= \frac{1}{(1-\frac{1}{q})^{m+1} (q-1)^{m-1} }\\
            &=  \frac{1}{  (1-\frac{1}{q})^{2} (1-\frac{1}{q})^{m-1} (q-1)^{m-1} }\\
            &=  \frac{1}{(1-\frac{1}{q})^{2} (q-2+ \frac{1}{q})^{m-1} }\\
            & < 3.
      \end{align*}
Therefore, $$ \Big (\frac{\lvert H \rvert}{\lvert  H_1 \rvert} \Big)^2 <  3\, \lvert H \rvert \implies \frac{\lvert H \rvert}{\lvert  H_1 \rvert}  <  2\, \sqrt{ \lvert H \rvert}. $$
\item[1.2] $A_m(q)$; $m \geq 1$, $q=2$ (Method 2)

\vspace{0.1cm}
The finite simple group $A_{m}(2)$ is of order $ 2^{\frac{m(m+1)}{2} }  \prod_{i=1}^{m} (2^{i+1}-1)/(q-1,m+1)$ and is isomorphic to \emph{projective special linear group} ${\rm PSL}_{m+1}(2)$ or ${\rm L}_{m+1}(2).$ It has a maximal subgroup of index $(2^{m+1}-1)$ (see \cite{KL}, p.  175). Let $H_1$ be one such maximal subgroup of $A_m(2)$. Then,
\begin{align*}
        \lvert H_{1} \rvert &= \frac{ \vert A_m(2) \vert }{(2^{m+1}-1)}\\
                      &= \frac{ 2^{\frac{m(m+1)}{2} }  \prod_{i=1}^{m} (2^{i+1}-1) }{ (2^{m+1}-1) }.
\end{align*}     
Let $H_2$ be the Sylow $2$-subgroup of $H_1.$  Then, $H_2$ has order $2^{\frac{m(m+1)}{2}}$ and $\vert H_{2} \vert^2 < \vert A_m(2) \vert. $ This implies that
$$ \frac{ \lvert H_1 \rvert }{\lvert H_2 \rvert }=  \prod_{i=1}^{m-1} (2^{i+1}-1).$$

Consider,
\begin{align*}
    \frac{\Big(\frac{ \lvert H_1 \rvert }{\lvert H_2 \rvert } \Big)^2} { \vert A_m(2) \vert } 
    &= \frac{\big(\prod_{i=1}^{m-1} (2^{i+1}-1) \big)^2} {2^{\frac{m(m+1)}{2} }  \prod_{i=1}^{m} (2^{i+1}-1) }\\
    &= \frac{\prod_{i=1}^{m-1} (2^{i+1}-1)} {  2^{\frac{m(m+1)}{2} (2^{m+1}-1)}}\\
    &< \frac{2^{\frac{m(m+1)}{2}-1}} {  2^{\frac{m(m+1)}{2} (2^{m+1}-1)}}\\
    &< 1. 
\end{align*}
Thus, we get
$$ \Big(\frac{ \lvert H_1 \rvert }{\lvert H_2 \rvert } \Big)^2 <  \lvert A_m(2) \rvert.$$
Now,
\begin{align*}
     \Big(\frac{ \vert A_m(2) \vert  }{\lvert H_1 \rvert } \Big)^2 &= (2^{m+1}-1)^2\\
                                                    &< 2^{\frac{m(m+1)}{2} }  \prod_{i=1}^{m} (2^{i+1}-1)\\
                                                    &= \lvert A_m(q) \rvert.
\end{align*}
\end{enumerate}

\subsubsection{A Classical Groups of Unitary Type:}
\begin{enumerate}
    \item[1.1]  $H= ~^2A_m(q^2)$; $m \geq 2$, $q>2$ (Method 1)
            
\vspace{0.1cm}
The finite simple group $^2A_m(q^2)$ is isomorphic to the  \emph{projective special unitary group} ${\rm PSU}_{m+1}(q)$. The group ${\rm PSU}_{m+1}(q)$ is the group obtain by taking special unitary group ${\rm SU}_{m+1}(q)$ and quotienting it by its center, i.e. $^2A_m(q^2) \cong \frac{{\rm SU}_{m+1}(q)}{Z(\rm{SU}_{m+1}(q))}$ (see \cite{WLS}, p. 66). It is known that (see \cite{ASH}, p.  252) the order of $^2A_m(q^2)$ is,
\allowdisplaybreaks
        $$\lvert H \rvert= \frac{q^{\frac{m(m+1)}{2} } }{(q+1,m+1)}  \prod_{i=1}^{m} (q^{i+1}-(-1)^{i+1}).$$
        
        Let $H_2$ be the Sylow $p$-subgroup of $^2A_m(q^2)$,  then $\lvert H_2 \rvert= q^{\frac{m(m+1)}{2}}$ and $ \lvert H_2 \rvert^2  \leq \lvert H \rvert.$ Let $H_1$ be the Borel subgroup of $H$ of order (see \cite{WLS}, \cite{CA}),
        %which is the semidirect product of $H_2$ and the diagonal subgroup $T$ of ${\rm SU}_{m+1}(q)$

        \begin{center}
         $\vert H_1 \vert=\frac{q^{\frac{m(m+1)}{2} }}{(q+1,m+1)}(q-1)^{\lfloor m/2 \rfloor} (q+1)^{\lceil \frac{m-1}{2} \rceil}$.
        \end{center}
        
        Consider,
        \begin{align*}
            \Big(\frac{\lvert H_1 \rvert}{\lvert H_2 \rvert} \Big)^2 
            &= \Big(\frac{(q-1)^{\lfloor m/2 \rfloor} (q+1)^{\lceil \frac{m-1}{2} \rceil} }{(q+1,m+1)}\Big)^2\\
            &\leq \frac{(q-1)^{m} (q+1)^{m}}{(q+1,m+1)^2}\\
            &= \frac{(q^{2}-1)^{m}}{(q+1,m+1)^2} \\
            & \leq \lvert H_1 \rvert.
        \end{align*}
Thus,
$$ \Big (\frac{\lvert H_1 \rvert}{\lvert H_2 \rvert} \Big)^2 \leq \lvert H_1 \rvert < \lvert H \rvert. $$
Also, 
\begin{align*}
    \frac{\lvert H \rvert}{\lvert H_1 \rvert^2} 
        &=\frac{  \frac{q^{\frac{m(m+1)}{2} } }{(q+1,m+1)}  \prod_{i=1}^{m} (q^{i+1}-(-1)^{i+1}) }{\Big(\frac{ (q-1)^{\lfloor m/2 \rfloor} (q+1)^{\lceil \frac{m-1}{2} \rceil} q^{\frac{m(m+1)}{2} }}{(q+1,m+1)  } \Big)^2 }\\
        &= \frac{(q+1,m+1) \prod_{i=1}^{m} (q^{i+1}-(-1)^{i+1}) }{  (q-1)^{m} (q+1)^{m-1} q^{\frac{m(m+1)}{2} } }\\
        &< \frac{ (q+1,m+1) q^{\frac{(m+1)(m+2)}{2}-1}}{ (q-1)^{m} (q+1)^{m-1} q^{{\frac{m^{2}+m}{2}}} }      \quad \quad \quad ( \text{by Remark}\, \ref{Remark-2}) \\
        &= \frac{(q+1,m+1) q^{\frac{(m+1)(m+2)}{2}-1}}{  (q-1) (q^{2}-1)^{m-1} q^{\frac{m^{2}+m}{2}}} \\
        &= \frac{(q+1,m+1)q^{m}}{(q-1)(q^{2}-1)^{m-1}}\\
        &< 2\, \frac{q}{\big(q-\frac{1}{q}\big)^{m-1}}\\
        &< 4.
        \end{align*}
Thus,
$$\Big (\frac{\lvert H \rvert}{\lvert H_1 \rvert} \Big )^2 \leq 4 \lvert H \rvert \implies \frac{\lvert H \rvert}{\lvert H_1 \rvert} \leq 2 \sqrt{\lvert H \rvert}$$

\item[1.2]  $H=~^2A_m(q^2)$; $m \geq 2$, $q=2$ (Method 2)

\vspace{0.1cm}
The finite simple group $^2A_m(2^2)$ is of order $ 2^{\frac{m(m+1)}{2} } \prod_{i=1}^{m} (2^{i+1}-(-1)^{i+1})/{(3,m+1)}$ and is isomorphic to \emph{projective special unitary group} ${\rm PSU}_{m+1}(2)$ or ${\rm U}_{m+1}(q)$. The group ${\rm U}_{m+1}(q)$ has a maximal subgroup of index ${\frac{ (2^{m+1}- (-1)^{m+1})(2^{m}-(-1)^{m})}{3} }$ when $6 \nmid (m-1)$ and of index $\frac{2^{m}(2^{m+1}-1)}{3},$ when $6 \mid (m-1)$  (see \cite{KL}, p. 175) .

 \begin{enumerate}
  \item [(Case 1)] $6 \nmid (m-1)$

\vspace{0.1cm}
Let $H_1$ be corresponding maximal subgroup of $^2A_m(2^2)$ whose index is $${\frac{ (2^{m+1}- (-1)^{m+1})(2^{m}-(-1)^{m})}{3} }$$ in $^2A_m(2^2).$ Then, the order of $H_1$ is,
\begin{align*}
        \lvert H_1 \rvert &= \frac{ \vert  ~^2A_m(2^2) \vert }{{\frac{ (2^{m+1}- (-1)^{m+1})(2^{m}-(-1)^{m})}{3} }}\\
                      &= \frac{3}{(3,m+1)} \frac{ {2^{\frac{m(m+1)}{2}}}  \prod_{i=1}^{m} (2^{i+1}-(-1)^{i+1})} {(2^{m+1}- (-1)^{m+1})(2^{m}-(-1)^{m})}.
\end{align*}     
Let $H_2$ be the Sylow 2-subgroup of $H_1.$  Then, $\lvert H_2 \rvert = 2^{\frac{m(m+1)}{2}}$ and $\lvert  H_{2} \rvert ^2 < \lvert ~^2A_m(2^2) \rvert. $ Also,
$$ \frac{ \lvert H_1 \rvert }{\lvert H_2 \rvert }=  \frac{3}{(3,m+1)} \frac{  \prod_{i=1}^{m} (2^{i+1}-(-1)^{i+1})} {(2^{m+1}- (-1)^{m+1})(2^{m}-(-1)^{m})}.$$
Consider,
\begin{align*}
    \frac{\Big(\frac{ \lvert H_1 \rvert }{\lvert H_2 \rvert } \Big)^2} { \vert ~^2A_m(2^2) \vert } 
     &= \frac{ \Big(\frac{3}{(3,m+1)} \frac{  \prod_{i=1}^{m} (2^{i+1}-(-1)^{i+1})} {(2^{m+1}- (-1)^{m+1})(2^{m}-(-1)^{m})}\Big)^2}  {\frac{2^{\frac{m(m+1)}{2} } }{(3,m+1)}  \prod_{i=1}^{m} (2^{i+1}-(-1)^{i+1})}\\
     &= \frac{9}{(3,m+1)} \frac{\prod_{i=1}^{m} (2^{i+1}-(-1)^{i+1})}  {2^{\frac{m(m+1)}{2} } ((2^{m+1}- (-1)^{m+1})(2^{m}-(-1)^{m}))^2 }\\
    &< \frac{9}{(3,m+1)} \frac{2^{\frac{(m+1)(m+2)}{2}-1 }}  {2^{\frac{m(m+1)}{2} } ((2^{m+1}- (-1)^{m+1})(2^{m}-(-1)^{m}))^2 } \quad \quad \quad ( \text{by Remark}\, \ref{Remark-2})\\
    &= \frac{9}{(3,m+1)}  \frac{2^{m}}{ ((2^{m+1}- (-1)^{m+1})(2^{m}-(-1)^{m}))^2}\\
    &< \frac{9}{(3,m+1)}  \frac{1}{ \Big(\big(2^{\lfloor \frac{m+1}{2}\rfloor }- \frac{(-1)^{m+1}}{2^{\lfloor \frac{m}{2} +1 \rfloor }}\big)(2^{m}-(-1)^{m})\Big)^2}\\
    &< 1.
 \end{align*}
 Therefore,
 $$     \Big(\frac{ \lvert H_1 \rvert }{\lvert H_2 \rvert } \Big)^2 <  \rvert~^2A_m(2^2)\lvert $$
 Now,
 \begin{align*}
     \Big(\frac{ \vert ~^2A_m(2^2) \vert  }{\lvert H_1 \rvert }\Big)^2
                                                    &= \Big({\frac{ (2^{m+1}- (-1)^{m+1})(2^{m}-(-1)^{m})}{3} }\Big)^2\\
                                                    &< \frac{2^{\frac{m(m+1)}{2} } }{(3,m+1)}  \prod_{i=1}^{m} (2^{i+1}-(-1)^{i+1}) \\      
                                                    &=  \lvert ~^2A_m(2^2) \rvert.
\end{align*}
\item[(Case 2)] \textbf{ $6 \vert (m-1)$} (i.e. $m \geq 7$)
   
   \vspace{0.1cm}
In this case, as we know that the group $^2A_m(q^2)$ has a maximal subgroup of index $ \frac{2^{m}(2^{m+1}-1)}{3}.$ Let $H_1$ be one such maximal subgroup. Then,
\begin{align*}
        \lvert H_1 \rvert &= \frac{ \vert  ~^2A_m(2^2) \vert }{\frac{2^{m}(2^{m+1}-1)}{3}}\\
                      &= \frac{\frac{2^{\frac{m(m+1)}{2} } }{(3,m+1)}  \prod_{i=1}^{m} (2^{i+1}-(-1)^{i+1})}{\frac{2^{m}(2^{m+1}-1)}{3}}\\
                      &= \frac{3}{(3,m+1)} {2^{\frac{m(m-1)}{2}}}  \prod_{i=1}^{m-1} (2^{i+1}-(-1)^{i+1}).
\end{align*}     
Let $H_2$ be the Sylow $2$-subgroup of $H_1,$ then $ H_2 $ has order $2^{\frac{m(m-1)}{2}}$ and $\lvert  H_{2} \rvert^2 < \lvert ~^2A_m(2^2) \rvert. $ Also,
$$ \frac{ \lvert H_1 \rvert }{\lvert H_2 \rvert }=   \frac{3}{(3,m+1)} \prod_{i=1}^{m-1} (2^{i+1}-(-1)^{i+1}).$$
Consider,
\begin{align*}
    \frac{\Big(\frac{ \lvert H_1 \rvert }{\lvert H_2 \rvert }\Big)^2} { \vert ~^2A_m(2^2) \vert } 
     &= \frac{\Big(  \frac{3}{(3,m+1)} \prod_{i=1}^{m-1} (2^{i+1}-(-1)^{i+1})\Big)^2}  {\frac{2^{\frac{m(m+1)}{2} } }{(3,m+1)}  \prod_{i=1}^{m} (2^{i+1}-(-1)^{i+1})}\\
     &= \frac{9}{(3,m+1)} \frac{\prod_{i=1}^{m-1} (2^{i+1}-(-1)^{i+1})}  {2^{\frac{m(m+1)}{2} } (2^{m+1}- (-1)^{m+1}) }\\
    &< \frac{9}{(3,m+1)} \frac{2^{\frac{m^2+m-2}{2}}}  {2^{\frac{m^{2}+m}{2} } (2^{m+1}- (-1)^{m+1})} \quad \quad \quad ( \text{by Remark}\, \ref{Remark-2})\\
    &= \frac{9}{(3,m+1)}  \frac{1}{2\,(2^{m+1}- (-1)^{m+1}) }\\
    &< 1. \quad \quad \quad \quad \quad \quad \quad \quad \quad \quad \quad \quad \quad \quad \quad \quad \quad \quad (\text{since} \, {m \geq 7})
\end{align*}
This implies that,
$$   \Big(\frac{ \lvert H_1 \rvert }{\lvert H_2 \rvert }\Big)^2 < \lvert~^2A_m(2^2) \rvert.$$
Now,
\begin{align*}
     \Big(\frac{ \lvert ~^2A_m(2^2) \rvert  }{\lvert H_1 \rvert }\Big)^2
                                                    &= \Big(\frac{2^{m}(2^{m+1}-1)}{3}\Big)^2\\
                                                    &= \frac{ 2^{2m}(2^{m+1}-1)^2}{9}\\         &< \frac{2^{\frac{m(m+1)}{2} } }{(3,m+1)}  \prod_{i=1}^{m} (2^{i+1}-(-1)^{i+1}) \\ 
                                                    &=  \lvert ~^2A_m(2^2) \rvert.
\end{align*}
\end{enumerate}
\end{enumerate}

\subsubsection{Classical Groups of Symplectic Linear Type: }
\begin{enumerate}
    \item[1.1]  $H= C_{m}(q)$; $m \geq 2$, $q>2 $ (Method 1)
     
     \vspace{0.1cm}   
The order of the finite simple group $C_{m}(q)$ could be found in (see \cite{ASH}, p. 252) and it is,
        \begin{center}
            $\lvert H \rvert= \frac{q^{m^2 }  \prod_{i=1}^{m} (q^{2i}-1)  }{(2,q-1)}.$
        \end{center}
        Let $H_2$ be the Sylow $p$-subgroup of $H$, then $\lvert H_2 \rvert= q^{m^2}$ and $ \lvert H_2 \rvert^2  \leq \lvert H \rvert.$ Let $H_1$ be the normalizer of $H_2$ in $H$ then the order of $H_1$  (see \cite{AI}, p. 3) is,  
        $$\lvert H_1 \rvert=\frac{q^{m^2}}{(2,q-1)} (q-1)^m. $$
        
          Consider,
        \begin{align*}
            \frac{\lvert H_1 \rvert}{\lvert H_2 \rvert^2} 
            &= \frac{\frac{q^{m^2}}{(2,q-1)} (q-1)^m }{q^{2m^2} }\\
            & = \frac{1}{(2,q-1)} \frac{(q-1)^m }{q^{m^2} }\\
            & < \frac{1}{(2,q-1)} \frac{q^m }{q^{m^2} }\\
            & < 1.
        \end{align*}
Thus,
$$  \Big (\frac{\lvert H_1 \rvert}{\lvert H_2 \rvert} \Big )^2 < \vert H_1 \vert < \lvert H \rvert. $$
        Also, 
      \begin{align*}
      \frac{\lvert H \rvert}{\lvert H_1 \rvert^{2}} 
            &=\frac{\frac{q^{m^2 }  \prod_{i=1}^{m} (q^{2i}-1)  }{(2,q-1)}}{( \frac{q^{m^2}}{(2,q-1)} (q-1)^m)^2 }\\
            &= \frac{ (2,q-1) \prod_{i=1}^{m} (q^{2i}-1)  }{ q^{m^2} (q-1)^{2m}}\\
            & < \frac{(2,q-1)}{q^{m^2}} \frac{\prod_{i=1}^{m} q^{2i}}{(q-1)^{2m} }  \\
            &= (2,q-1) \frac{q^{m+m^2}}{q^{m^2} (q-1)^{2m}} \\
            &= (2,q-1)  (\frac{q}{(q-1)^2})^{m}\\
            &\leq 2 \quad \quad \quad ( \text{by Remark}\, \ref{Remark-1})
 \end{align*}
 Therefore,
 $$  \Big(\frac{\lvert H \rvert}{\vert H_1 \vert} \Big)^2 \leq 2 \, \lvert H \rvert \implies \frac{\lvert H \rvert}{\vert H_1 \vert}  < 2 \, \sqrt{\lvert H \rvert}.$$
 \item[1.2]  $H= C_{m}(q)$; $m \geq 2$, $q=2$ (Method 2)
 
\vspace{0.1cm}
The simple group $C_m(q)$ (or ${\rm PSp}_{2m}(q)$) has order $2^{m^2 }  \prod_{i=1}^{m} (2^{2i}-1)$. It is known that the group  ${ \rm PSp}_{2m}(q)$ has a maximal subgroup of index $ 2^{m-1} (2^{m}-1)$ (see \cite{KL}, p.  175). Let $H_1$ be one such subgroup, then the order of $H_1$ is,
\begin{align*}
    \lvert H_1 \rvert &= \frac{ \vert C_m(2) \vert }{2^{m-1} (2^{m}-1)}\\
                      &= \frac{ 2^{m^2 }  \prod_{i=1}^{m} (2^{2i}-1)}{ 2^{m-1} (2^{m}-1) }\\
                      &= \frac{ 2^{m^{2} -m+1} (2^{2m} -1)  \prod_{i=1}^{m-1} (2^{2i}-1)}{ (2^{m}-1) }\\
                      &=  2^{m^{2} -m+1} (2^{m} +1)  \prod_{i=1}^{m-1} (2^{2i}-1).
\end{align*}     
Let $H_2$ be the Sylow $2$-subgroup of $H_1.$  Then, the order of $ H_2 $ is $2^{m^{2} -m+1}.$ Notice that $\lvert H_{2} \rvert^2 < \lvert C_m(2) \rvert.$ Thus,
$$ \frac{ \lvert H_1 \rvert }{\lvert H_2 \rvert }= (2^{m} +1)  \prod_{i=1}^{m-1} (2^{2i}-1).$$
Consider,
\begin{align*}
    \frac{\Big(\frac{ \lvert H_1 \rvert }{\lvert H_2 \rvert }\Big)^2} { \lvert C_m(2) \rvert } &= \frac{((2^{m} +1)  \prod_{i=1}^{m-1} (2^{2i}-1))^2} {2^{m^2 }  \prod_{i=1}^{m} (2^{2i}-1) }\\
    &=\frac{(2^{m} +1)^2  (\prod_{i=1}^{m-1} (2^{2i}-1))^2} {2^{m^2 } (2^{2m}-1)  \prod_{i=1}^{m-1} (2^{2i}-1) }\\
    &=\frac{(2^{m} +1)  \prod_{i=1}^{m-1} (2^{2i}-1)} {2^{m^2 } (2^{m}-1)  }\\
    &< \frac{(2^{m} +1)  \prod_{i=1}^{m-1} 2^{2i}} {2^{m^2 } (2^{m}-1)  }\\
    &= \frac{(2^{m} +1)  2^{m(m-1)}} {2^{m^2 } (2^{m}-1)  }\\
    &= \frac{(1+2^{-m})} { (2^{m}-1)  }\\
    &< 1.
\end{align*}
Therefore,
$$ \Big(\frac{ \lvert H_1 \rvert }{\lvert H_2 \rvert }\Big)^2 <   \lvert C_m(2) \rvert.$$

 Now,
\begin{align*}
     \Big(\frac{ \vert C_m(2) \vert  }{\lvert H_1 \rvert }\Big)^2 &= (2^{m-1} (2^{m}-1))^2\\
                                                    &= 2^{2m-2} (2^{m}-1)^2\\
                                                    &< 2^{m^{2}} (2^{m}-1) (2^{m}+1)\\
                                                    &< 2^{m^2} \prod_{i=1}^{m} (2^{2i}-1) \\
                                                    &=  \lvert C_m(2) \rvert.
\end{align*}
       
\end{enumerate}
\subsubsection{Classical Groups of Orthogonal Type: }

\begin{enumerate}
       \item[1.1] $H=B_{m}(q);$ $m \geq 3$ and $q$ odd (Method 1)
        
        \vspace{0.1cm}
         The finite simple group $B_{m}(q)$ has order (see \cite{ASH}, p. 252),
         %is isomorphic  to the \emph{orthogonal group} ${\rm O}_{2m+1}(q)$ in $2m+1$ dimension. Thus, the order of ${\rm O}_{2m+1}(q)$ or $B_{m}(q)$ is 
        \begin{center}
            $\lvert H \rvert= \frac{q^{m^2 }  \prod_{i=1}^{m} (q^{2i}-1)  }{(2,q-1)}.$
        \end{center}
        
        Let $H_2$ be the Sylow $p$-subgroup of $H$, then $\lvert H_2 \rvert= q^{m^2}$ and $ \lvert H_2 \rvert^2  \leq \lvert H \rvert.$ Let $H_1$ be the normalizer of $H_1$ in $H$ then its order is (see \cite{AI}, p. 3), 
        $$\lvert H_1\rvert=\frac{q^{m^2}}{(2,q-1)} (q-1)^m. $$
        
        Since, $ \lvert B_{m}(q) \rvert = \lvert C_{m}(q) \rvert$ and the order of the normalizer of $H_1$ in both the groups are also equal and thus, all the calculations will also work for $B_{m}(q).$

      \item[2.1] $H= D_{m}(q)$; $m \geq 4$, $q>2$ (Method 1)
      
        \vspace{0.1cm}
        The order of the finite simple group $D_{m}(q)$ could be found in (see \cite{ASH}, p. 252) and it is,
        \begin{center}
        $\vert H \vert = \frac{q^{{m(m-1)}} (q^{m}-1) \prod_{i=1}^{m-1} (q^{2i}-1)  }{(4, q^{m}-1)}.$
        \end{center}
        Let $H_2$ be the Sylow $p$-subgroup of $D_{m}(q)$,  then $\lvert H_2 \rvert = q^{{m(m-1)}}$ and $ \lvert H_2 \rvert^2  \leq \lvert H \rvert.$ Let $H_1$ be the Borel subgroup of $H$. It has order (see \cite{WLS}, \cite{CA}), $$\lvert H_1 \rvert =\frac{q^{{m(m-1)}} }{ (4, q^{m}-1)}  (q-1)^m.$$
    Consider,
    \begin{align*}
            \frac{\lvert H_1 \rvert}{\lvert H_2 \rvert^2} 
            &= \frac{ \frac{q^{{m(m-1)}} }{ (4, q^{m}-1)}  (q-1)^m }{q^{2m(m-1)} }\\
            & = \frac{1}{(4, q^{m}-1)} \frac{(q-1)^m }{q^{m(m-1)} }\\
            & < \frac{1}{(4, q^{m}-1)} \frac{q^m }{q^{m(m-1)} }\\
            & < 1.
        \end{align*}
Thus, $$ \Big( \frac{\lvert H_1 \rvert}{\lvert H_2 \rvert} \Big )^2 < \vert H_1  \vert < \lvert H \rvert $$
      Also, 
\allowdisplaybreaks
    \begin{align*}
      \frac{\lvert H \rvert}{\lvert H_1 \rvert^{2}} 
            &=\frac{\frac{q^{{m(m-1)}} (q^{m}-1) \prod_{i=1}^{m-1} (q^{2i}-1)  }{(4, q^{m}-1)}}{(  \frac{q^{{m(m-1)}} }{ (4, q^{m}-1)}  (q-1)^m )^2 }\\
            &= \frac{ (4, q^{m}-1) (q^{m}-1) \prod_{i=1}^{m} (q^{2i}-1)  }{ q^{{m(m-1)}}  (q-1)^{2m}}\\
            &< \frac{(4, q^{m}-1) (q^{m}-1) }{ q^{{m(m-1)}}} \frac{\prod_{i=1}^{m-1} q^{2i}}{(q-1)^{2m} }  \\
            &< \frac{(4, q^{m}-1) q^{m} }{ q^{{m(m-1)}}} \frac{ q^{{m(m-1)}} }{(q-1)^{2m} }  \\
            &=(4, q^{m}-1) \Big(\frac{q}{(q-1)^2}\Big)^{m} \\
            & \leq 4.  \quad \quad \quad ( \text{by Remark}\, \ref{Remark-1})
    \end{align*}
Thus, $$ \Big(\frac{\lvert H \rvert}{\vert H_1 \vert} \Big)^2 \leq  4\, \lvert H \rvert \implies \frac{\lvert H \rvert}{\vert H_1 \vert} \leq  2 \sqrt{ \lvert H \rvert}. $$
     \item[2.2] $H= D_{m}(q)$; $m \geq 4$, $q=2$ (Method 2)
     
     \vspace{0.1cm}
The simple group $D_m(q)$ (or ${\rm P\Omega}_{2m}^{+}(q)$) is of order $2^{{m(m-1)}} (2^{m}-1) \prod_{i=1}^{m-1} (2^{2i}-1)$. The group ${\rm P\Omega}_{2m}^{+}(2)$ has a maximal subgroup of index $2^{m-1} (2^{m}-1)$  (see \cite{KL}, p. 175). Let $H_1$ be a corresponding maximal subgroup of $D_m(2)$ whose index is $2^{m-1} (2^{m}-1).$ Then the order of $H_1$ is,
\begin{align*}
        \lvert H_1 \rvert &= \frac{ \vert D_m(2) \vert }{2^{m-1} (2^{m}-1)}\\
                      &= \frac{2^{{m(m-1)}} (2^{m}-1) \prod_{i=1}^{m-1} (2^{2i}-1)}{ 2^{m-1} (2^{m}-1) }\\
                      &=  2^{m^{2}-2m+1} \prod_{i=1}^{m-1} (2^{2i}-1).
\end{align*}     
Let $H_2$ be the Sylow 2-subgroup of $H_1,$  then $ H_2 $ has order $2^{m^{2}-2m+1}$ and $\lvert  H_{2} \rvert^2 < \lvert D_m(2) \rvert. $ Also,
$$ \frac{ \lvert H_1 \rvert }{\lvert H_2 \rvert }= \prod_{i=1}^{m-1} (2^{2i}-1).$$
Consider,
\begin{align*}
    \frac{\Big(\frac{ \lvert H_1 \rvert }{\lvert H_2 \rvert }\Big)^2} { \vert D_m(2) \vert } &= \frac{(\prod_{i=1}^{m-1} (2^{2i}-1))^2}  {2^{{m(m-1)}} (2^{m}-1) \prod_{i=1}^{m-1} (2^{2i}-1) }\\
     &= \frac{\prod_{i=1}^{m-1} (2^{2i}-1)}  {2^{{m(m-1)}} (2^{m}-1) }\\
     &<  \frac{\prod_{i=1}^{m-1} 2^{2i}}  {2^{{m(m-1)}} (2^{m}-1) }\\
     &=  \frac{2^{m(m-1)} } {2^{{m(m-1)}} (2^{m}-1) }\\
     &= \frac{1 }{(2^{m}-1)}\\
     &< 1
\end{align*}
Thus,
$$ \Big(\frac{ \lvert H_1 \rvert }{\lvert H_2 \rvert }\Big)^2 <  \lvert D_m(2) \rvert. $$
 Now,
 \begin{align*}
     \Big(\frac{ \lvert D_m(2) \rvert  }{\lvert H_1 \rvert }\Big)^2 &= 2^{2(m-2)} (2^{m}-1)^2\\
                                                    &<2^{{m(m-1)}} (2^{m}-1) \prod_{i=1}^{m-1} (2^{2i}-1)  \\
                                                    &=  \lvert D_m(2) \rvert.
\end{align*}
    
     \item[3.1] $H= ~^2D_m(q^2)$; $m \geq 4$, $q>2$ (Method 1)
            
            \vspace{0.1cm}
        The order of the finite simple group $^2D_m(q^2)$ could be found in (see \cite{ASH}, p. 252) and it is,
        %is isomorphic to \emph{general orthogonal group} $O_{2n}^{-}(q).$ Thus, it is well known that the order of $^2D_m(q^2)$ is (see \cite{WLS})
        \begin{center}
        $\lvert H \rvert=\frac{q^{m(m-1)}(q^{m}+1) }{(4, q^{m}+1)} \prod_{i=1}^{m-1} (q^{2i}-1).$  
        \end{center}
        Let $H_2$ be the Sylow $p$-subgroup of $^2D_m(q^2)$, then $\lvert H_2 \rvert= q^{m(m-1)}$ and $ \lvert H_2 \rvert^2  \leq \lvert H \rvert.$ Let $H_1$ to be the Borel subgroup of $H$. The order of $H_1$ is (see \cite{WLS}, \cite{CA}). Thus,
        
        \begin{center}
         $\vert B\vert=\frac{q^{{m(m-1)}} }{ (4, q^{n}+1)}  (q-1)^m$.
        \end{center}
         Consider, 
        \begin{align*}
            \frac{\lvert H_1 \rvert}{\lvert H_2 \rvert^2} 
            &= \frac{ \frac{q^{{m(m-1)}} }{ (4, q^{n}+1)}  (q-1)^m  }{q^{2m(m-1)} }\\
            & = \frac{1}{(4, q^{n}+1)} \frac{(q-1)^m }{q^{m(m-1)} }\\
            &< \frac{1}{(4, q^{n}+1)} \frac{q^m }{q^{m(m-1)} }\\
            & \leq 1.
        \end{align*}
        Thus, we have $$  \Big (\frac{\lvert H_1 \rvert}{\lvert H_2 \rvert} \Big)^2 < \vert H_1 \vert <\lvert H \rvert.$$
        Also, 
    \begin{align*}
      \frac{\lvert H \rvert}{\lvert H_1 \rvert^{2}} 
            &=\frac{\frac{q^{m(m-1)}(q^{m}+1) }{(4, q^{m}+1)} \prod_{i=1}^{m-1} (q^{2i}-1)}{( \frac{q^{{m(m-1)}} }{ (4, q^{m}+1)}  (q-1)^m )^2 }\\
            &= \frac{ (4, q^{m}+1) (q^{m}+1) \prod_{i=1}^{m-1} (q^{2i}-1)}{ q^{m(m-1)} (q-1)^{2m}}\\
            &< \frac{(4, q^{m}+1) (q^{m}+1) }{q^{m(m-1)}} \frac{\prod_{i=1}^{m-1} q^{2i}}{(q-1)^{2m} }  \\
            &= (4, q^{m}+1) (q^{m}+1) \frac{q^{m(m-1)}}{q^{m(m-1)} (q-1)^{2m}} \\
            &<  \frac{4 (2q^{m})}{(q-1)^{2m}} \\
            & < 8  \quad \quad \quad \quad \quad \quad \quad \quad \quad \quad \quad \quad \quad \quad \quad ( \text{by Remark}\, \ref{Remark-1}) 
    \end{align*}
This implies that,
$$  \Big (\frac{\lvert H \rvert}{\vert H_1 \vert} \Big )^2 < 8\,\lvert H \rvert \implies \frac{\lvert H \rvert}{\vert H_1 \vert} < 3\,\sqrt{\lvert H \rvert}.$$
     
 \item[3.2] $H= ~^2D_m(q^2)$; $m \geq 4$, $q=2$ (Method 2)
      
\vspace{0.1cm}
We know that the group $^2D_m(2^2)$(or ${\rm P\Omega}_{2m}^{-}(2)$) is of order $2^{m(m-1)}(2^{m}+1) \prod_{i=1}^{m-1} (2^{2i}-1)/(4, 2^{m}+1).$ The group ${\rm P\Omega}_{2m}^{-}(2)$ has a maximal subgroup of index $(2^{m}+1) (2^{m-1}-1)$ (see \cite{KL}, p. 175). Let $H_1$ be one such subgroup then its order is,
\begin{align*}
\lvert H_1 \rvert &= \frac{ \vert  ~^2D_m(2^2) \vert }{(2^{m}+1) (2^{m-1}-1)}\\
                      &= \frac{\frac{2^{m(m-1)}(2^{m}+1) }{(4, 2^{m}+1)} \prod_{i=1}^{m-1} (2^{2i}-1) }{(2^{m}+1) (2^{m-1}-1) }\\
                      &=  2^{m(m-1)} (2^{m-1}+1)\prod_{i=1}^{m-2} (2^{2i}-1).
\end{align*}
Let $H_2$ be the Sylow $2$-subgroup of $H_1,$ then $\lvert H_2 \rvert= 2^{m(m-1)}$ and $\vert  H_{2} \vert^2 < \vert ~^2D_m(2^2) \vert. $ Thus,
$$ \frac{ \lvert H_1 \rvert }{\lvert H_2 \rvert }= (2^{m-1}+1)\prod_{i=1}^{m-2} (2^{2i}-1).$$
Consider,
\begin{align*}
    \frac{\Big(\frac{ \lvert H_1 \rvert }{\lvert H_2 \rvert }\Big)^2} { \lvert ~^2D_m(2^2) \rvert } &= \frac{((2^{m-1}+1)\prod_{i=1}^{m-2} (2^{2i}-1))^2}  {\frac{2^{m(m-1)}(2^{m}+1) }{(4, 2^{m}+1)} \prod_{i=1}^{m-1} (2^{2i}-1)  }\\
     &= \frac{(2^{m-1}+1)^2 \prod_{i=1}^{m-2} (2^{2i}-1)}  {2^{{m(m-1)}} (2^{2m-2}-1) (2^{m}+1) }\\
     &= \frac{(2^{m-1}+1) \prod_{i=1}^{m-2} (2^{2i}-1)}  {2^{{m(m-1)}} (2^{m-1}-1) (2^{m}+1) }\\
     &< \frac{(2^{m-1}+1) \prod_{i=1}^{m-2} 2^{2i}}  {2^{{m(m-1)}} (2^{m-1}-1) (2^{m}+1) }\\
     &= \frac{(2^{m-1}+1) 2^{(m-2)(m-1)} }  {2^{{m(m-1)}} (2^{m-1}-1) (2^{m}+1) }\\
     &= \frac{(2^{m-1}+1) }{(2^{m-1}-1) (2^{m}+1)} \cdot \frac{  2^{(m-2)(m-1)} }{2^{{m(m-1)}}}\\
     &< 1.\\
\Big(\frac{ \lvert H_1 \rvert }{\lvert H_2 \rvert }\Big)^2 &<  \lvert ~^2D_m(2^2)\rvert.
\end{align*} \\
 Now,
 \begin{align*}
     \Big(\frac{ \vert ~^2D_m(2^2) \vert  }{\lvert H_1 \rvert }\Big)^2 
                                                    &= ((2^{m}+1) (2^{m-1}-1))^2\\
                                                    &< 2^{m(m-1)} (2^{m}+1) (2^{m-1}-1) (2^{m-1}+1) \prod_{i=1}^{m-2} (2^{2i}-1) \\
                                                    &<2^{m(m-1)}(2^{m}+1) \prod_{i=1}^{m-1} (2^{2i}-1)  \\
                                                    &=  \lvert ~^2D_m(2^2) \rvert.
\end{align*}
\end{enumerate}

\subsection{Exceptional Group of Lie Type}
\begin{enumerate}
   \item[(1)]  $H= G_2(q)$; $q\geq 3$ (Method 1)
        
        \vspace{0.1cm}
        The group $G_2(q)$ is simple for all $q\geq 3$. It has order (see \cite{ASH}, p. 252),
        $$\vert G_2(q) \vert = q^{6} (q^{6}-1)(q^{2}-1).$$ 
        Thus, $G_2(q)$  has a Sylow $p$-subgroup $H_2$ of order $q^{6}$ and $ \lvert H_2 \rvert^2  \leq \rvert G_2(q) \lvert.$ Also,  it has the Borel subgroup $H_1$ of order $q^{6}(q-1)^2$ (see \cite{WLS}, p. 124).
        
        Consider,
        \begin{align*}
            \frac{\lvert  H_1 \rvert}{\lvert H_2 \rvert^2} 
            &= \frac{ q^{6}(q-1)^2}{q^{12} }\\
            &= \frac{(q-1)^2}{q^{6} }\\
            & < 1. 
        \end{align*}
Therefore,
$$  \Big (\frac{\lvert  H_1 \rvert}{\lvert H_2 \rvert} \Big)^2 \leq \lvert  H_1 \rvert \leq \lvert H \rvert. $$
Now,
        \begin{align*}
         \frac{\lvert H \rvert}{\vert H_1 \vert^2} 
            &=\frac{ q^{6} (q^{6}-1)(q^{2}-1)}{( q^{6}(q-1)^2)^2 }\\
            & < \frac{q^{6+2}}{ q^{6}(q-1)^4}\\
            & < \frac{q^{2}}{(q-1)^4}\\
            & < 1  \quad \quad \quad \quad \quad \quad  \quad \quad \quad  \quad \quad \quad ( \text{by Remark}\, \ref{Remark-1})
        \end{align*} 
Thus,
$$ \Big(\frac{\vert H\vert}{\lvert  H_1 \rvert} \Big)^2 < \lvert H \rvert.$$
    \item[(2)]  $H=F_4(q)$ (Method 2)
        
        \vspace{0.1cm}
        The finite simple group $F_4(q)$ has order (see \cite{ASH}, p. 252), $$\vert  F_4(q) \vert = q^{24} (q^{12}-1)(q^{8}-1)(q^{6}-1)(q^{2}-1).$$ 
        
        It is known that (see \cite{WLS}, p.  156) the group $F_4(q)$ has a maximal subgroup $q^{1+14}: Sp_{6}(q).C_{q-1}$ of order $q^{24}(q^{6}-1)(q^4-1) (q^{2}-1) (q-1)$ say $H_1$. This subgroup has a Sylow $p$-subgroup say $H_2$ of order $q^{24}$ and $ \lvert H_2 \rvert^2  \leq \vert F_4(q) \vert.$ Now,
        \begin{align*}
            \frac{\vert H_1 \vert}{\lvert H_2 \rvert^2} 
            &= \frac{ q^{24} (q^{6}-1)(q^4 -1) (q^{2}-1) (q-1) }{q^{48} }\\
            &< \frac{q^{1+2+4+6} }{q^{24} }\\
            &<1.\\
        \end{align*}
Therefore, $$  \Big (\frac{\vert H_1 \vert}{\lvert H_2 \rvert} \Big)^2 < \vert H_1  \vert < \lvert H \rvert.$$
    Now, 
     \begin{align*}
         \frac{\lvert H \rvert}{\vert H_1 \vert^2} 
            &=\frac{q^{24} (q^{12}-1)(q^{8}-1)(q^{6}-1)(q^{2}-1)}{( q^{24} (q^{6}-1)(q^4 -1) (q^{2}-1) (q-1))^2 }\\
            &= \frac{ (q^{4} -1)^2 (q^8+q^4+1) (q^4+1) }{ q^{24} (q^{6}-1)(q^4 -1)^2 (q^{2}-1) (q-1)^2}\\
            &= (1+ \frac{1}{q^4} + \frac{1}{q^8}) (1+ \frac{1}{q^4}) \frac{1  }{ q^{12} (q^{6}-1) (q^{2}-1) (q-1)^2} \\
            &< 1.
            \end{align*}
Thus, we get
$$  \Big (\frac{\lvert H \rvert}{\vert H_1 \vert} \Big)^2 < \lvert H \rvert.$$            
    
    \item[(3)]  $H=E_6(q)$; $q>2$ (Method 1)
        
        \vspace{0.1cm}
        The group $E_{6}(q)$ is a finite simple group. The order of $H=E_6(q)$ is (see \cite{ASH}, p. 252),
        $$\vert  E_6(q) \vert = \frac{q^{36}}{(3,q-1)} (q^{12}-1)(q^{9}-1)(q^{8}-1)(q^{6}-1)(q^{5}-1)(q^{2}-1).$$ Clearly, it has a Sylow $p$-subgroup $H_2$ of order $q^{36}$ and $ \lvert H_2 \rvert^2  \leq \vert E_6(q) \vert.$ Let $H_1$ be the Borel subgroup of $H$ then the order of $H_1$ is $q^{36} (q-1)^6$ (see \cite{WLS}, \cite{CA}).
        Consider, 
        \begin{align*}
            \frac{\lvert H_1 \rvert}{\lvert H_2 \rvert^2} 
            &= \frac{ q^{36} (q-1)^6}{q^{72}}\\
            &< \frac{q^6 }{q^{36} }\\
            &< 1.
        \end{align*}
Thus,
$$ \Big (\frac{\lvert H_1 \rvert}{\lvert H_2 \rvert}\Big )^2 < \lvert H_1 \rvert < \lvert H \rvert.$$
        Now,
        \begin{align*}
         \frac{\lvert H \rvert}{\lvert H_1 \rvert^2} 
            &=\frac{q^{36} (q^{12}-1)(q^{9}-1)(q^{8}-1)(q^{6}-1)(q^{5}-1)(q^{2}-1)}{(3,q-1) ( q^{36} (q-1)^8 )^2 }\\
            &< \frac{q^{12+9+8+6+5+2}}{(3,q-1) q^{36} (q-1)^{16} }   \\
            &= \frac{q^{42}}{ (3,q-1) q^{36} (q-1)^{16} } \\
            &= \frac{ q^{6}}{ (3,q-1) (q-1)^{16} } \\
            &=\frac{1}{(3,q-1) (q-1)^{4}} \frac{ q^{6}}{(q-1)^{12} } \\
            &< \frac{ q^{6}}{(q-1)^{12} } \\
            &<1. \quad \quad \quad  \quad \quad \quad   \quad \quad \quad (\text{by Remark}\, \ref{Remark-1})\\
        \end{align*}
This implies that,
$$ \Big (\frac{\lvert H \rvert}{\lvert H_1 \rvert} \Big )^2 \leq  \lvert H \rvert.$$
        
Notice that, the group $E_6(2)$ is of constant order. However, we can use Method 2 to reduce the constants $b_1$, $b_2$ to $1$. By taking $H_1$ to be maximal subgroup of order (see \cite{ATLAS}) $2^{36}\cdot3^{3}\cdot5\cdot7\cdot31$ and $H_2$ to be its Sylow $2$-subgroup of order $2^{36}$.
    \vspace{0.1cm}

    \item[(4)] $H= ~^2E_6(q)$ (Method 1)
            
            \vspace{0.1cm}
           The group $^2E_6(q)$ is simple for all $q$ and it has order (see \cite{ASH}, p. 252), $$\vert  ~^2E_6(q) \vert = \frac{q^{36}(q^2-1)(q^5+1)(q^6-1)(q^8-1)(q^9+1)(q^{12}-1)}{(3,q+1)}.$$ 
            The group $~^2E_6(q)$ has a Sylow $p$-subgroup $H_2$ of order $q^{36}$ and $ \lvert H_2 \rvert^2  \leq \vert H \vert.$ It is well known that the order of the Borel subgroup is a product of the order of the Sylow $p$-subgroup and the maximal torus. The order of the maximal torus in $H$ is $(q+1)^6$(see \cite{WLS}, p. 173). Thus, the order of the Borel subgroup $H_1$ of $H$ is $q^{36} (q+1)^6$. 
            
            Consider, 
        \begin{align*}
            \frac{\lvert  H_1 \rvert}{\lvert H_2 \rvert^2} 
            &= \frac{ q^{36} (q+1)^6}{q^{72} }\\
            &< \frac{(q+1)^6}{q^{36} }\\
            &< \frac{(q+q)^6 }{q^{36} }\\
            & < 1.
        \end{align*}
Thus, $$ \Big(\frac{\vert H_1 \vert}{\lvert H_2 \rvert} \Big)^2 <  \lvert H_1 \rvert <  \lvert H \rvert.$$
    Also,    
     \begin{align*}
         \frac{\lvert H \rvert}{\vert H_1 \vert^2} 
            &=\frac{ q^{36} (q^{2}-1)(q^{5}+1)(q^{6}-1)(q^{8}-1)(q^{9}+1)(q^{12}-1)}{( q^{36}(q+1)^6 )^2 }\\
            &< \frac{(q^{5}+1)(q^9+1)}{q^8(q+1)^{12} } \\
            &< 1.
        \end{align*}
    Therefore,
    $$  \Big (\frac{\lvert H \rvert}{\lvert  H_1 \rvert} \Big)^2 < \lvert H \rvert.$$

    \item[(5)] $H = ~^3D_4(q)$ (Method 1)
    
    \vspace{0.1cm}
            The group $^3D_4(q)$ is simple for all $q$. It is known that the order of the group $^3D_4(q)$ is (see \cite{ASH}, p. 252), $$\vert  ~^3D_4(q) \vert= q^{12} (q^{8}+q^4+1)(q^{6}-1)(q^{2}-1).$$ Clearly, it has a Sylow $p$-subgroup $H_2$ of order $q^{12}$ and the Borel subgroup $H_1$ of order $q^{12}(q^3-1) (q-1)$ (see \cite{WLS}, p. 141). Also, notice that $ \lvert H_2 \rvert^2  \leq \lvert H \rvert.$
        
        Consider, 
        \begin{align*}
            \frac{\lvert  H_1 \rvert}{\lvert H_2 \rvert^2} 
            &= \frac{ q^{12} (q^3-1)(q-1)}{q^{24} }\\
            &= \frac{(q^3-1)(q-1) }{q^{12} }\\
            & < \frac{q^4 }{q^{12} }\\
            & < 1.
        \end{align*}
        
Thus, $$ \Big (\frac{\vert H_1 \vert}{\lvert H_2 \rvert} \Big)^2 < \vert H_1\vert < \lvert H \rvert.$$
     \begin{align*}
         \frac{\lvert H \rvert}{\vert H_1 \vert^2} 
            &=\frac{ q^{12}(q^{8}+q^4+1)(q^{6}-1)(q^{2}-1)}{( q^{12} (q^3-1)(q-1) )^2 }\\
            &=  \frac{(q^{8}+q^4+1)(q^3+1)(q+1)}{q^{12}(q^3-1)(q-1) }\\
            &= \frac{(1+\frac{1}{q^4}+\frac{1}{q^8})(1+\frac{1}{q^3})(1+\frac{1}{q})}{(q^3-1)(q-1)} \\
            & < 1.
        \end{align*}
Therefore,
$$  \Big (\frac{\lvert H \rvert}{\lvert  H_1 \rvert} \Big)^2 < \lvert H \rvert.$$
     \item[(6)] $H= E_7(q)$, $q>2$ (Method 1)
        
        \vspace{0.1cm}
        The simple group $E_7(q)$ has order (see \cite{ASH}, p. 252),
        $$\vert  E_7(q) \vert = \frac{q^{63}}{(2,q-1)} (q^{18}-1)(q^{14}-1)(q^{12}-1)(q^{10}-1)(q^{8}-1)(q^{6}-1)(q^{2}-1).$$ 
        Clearly, it has a Sylow $p$-subgroup $H_2$ of order $q^{63}$ and $ \lvert H_2 \rvert^2  \leq \lvert H \rvert.$ Let $H_1$ be the borel subgroup of $H$ then its order is (see \cite{WLS}, \cite{CA}),
        $$\lvert H_1 \rvert = q^{63} (q-1)^7.$$ 
        \allowdisplaybreaks
        Consider, 
        \begin{align*}
            \frac{\lvert H_1 \rvert}{\lvert H_2 \rvert^2} 
            &= \frac{ q^{63} (q-1)^7}{q^{126}}\\
            &< \frac{q^7 }{q^{63} }\\
            &< 1
        \end{align*}
Thus,
$$   \Big (\frac{\lvert H_1 \rvert}{\lvert H_2 \rvert} \Big )^2  < \lvert H_1 \rvert  < \lvert H \rvert.$$
        \begin{align*}
         \frac{\lvert H \rvert}{\lvert B \rvert^2} 
            &=\frac{\frac{q^{63}}{(2,q-1)} (q^{18}-1)(q^{14}-1)(q^{12}-1)(q^{10}-1)(q^{8}-1)(q^{6}-1)(q^{2}-1)}{ ( q^{63} (q-1)^7 )^2 }\\
            &< \frac{q^{18+14+12+10+8+6+2}}{(2,q-1) q^{63} (q-1)^{14} }   \\
            &= \frac{q^{70}}{ (2,q-1) q^{63} (q-1)^{14} } \\
            &= \frac{1}{(2,q-1)} \frac{q^{7}}{ (q-1)^{14}}\\
            &< 1.  \quad \quad \quad \quad \quad \quad \quad \quad \quad \quad \quad \quad \quad\quad \quad \quad  (q \neq 2). 
        \end{align*}
Therefore,
$$ \Big (\frac{\lvert H \rvert}{\lvert H_1 \rvert} \Big )^2 < \lvert H \rvert. $$
 Notice that, the group $E_7(2)$ is of constant order. However, we can use Method 2 to reduce the constants $b_1$, $b_2$ to $1$. By taking $H_1$ to be maximal subgroup of order (see \cite{ATLAS}) $2^{63}\cdot3^{4}\cdot7^{2}\cdot5$  and $H_2$ to be its Sylow $2$-subgroup of order $2^{63}$.
 
        %The group $E_7(2)$ is of constant order. However, we can use Method 2 to reduce the constants $b_1, b_2$ to $1$. Here, we take $H_1$ to be maximal subgroup whose order is $2^{63}*3^4*7^2*5$ and $H_2$ to be its Sylow 2-subgroup.
        \vspace{0.1cm}

        \item[(7)] $H= E_8(q)$, $q>2$ (Method 1)
        
        \vspace{0.1cm}
        We know that the group $E_8(q)$ is simple for all $q$ and has order (see \cite{ASH}, p. 252),
        $$\vert H \vert = q^{120} (q^{30}-1)(q^{24}-1)(q^{20}-1)(q^{18}-1)(q^{14}-1)(q^{12}-1)(q^{8}-1)(q^{2}-1).$$ 
        The group $E_8(q)$ has a Sylow $p$-subgroup $H_2,$ then $\lvert H_2 \rvert = q^{120}$ and $ \lvert H_2 \rvert^2  \leq \vert H \vert.$ Let $H_1$ be the Borel subgroup then order $\lvert H_1 \rvert = q^{120} (q-1)^8 $ (see \cite{WLS}, p. 176).
        
        Consider, 
        \begin{align*}
            \frac{\lvert  H_1 \rvert}{\lvert H_2 \rvert^2} 
            &= \frac{ q^{120} (q-1)^8}{q^{240} }\\
            &< \frac{q^8 }{q^{120} }\\
            &< 1.
        \end{align*}
Threfore,
$$  \Big (\frac{\vert H_1 \vert}{\lvert H_2 \rvert} \Big )^2 < \vert H_1\vert < \lvert H \rvert.$$
       \begin{align*}
         \frac{\lvert H \rvert}{\vert H_1 \vert^2} 
            &=\frac{ q^{120} (q^{30}-1)(q^{24}-1)(q^{20}-1)(q^{18}-1)(q^{14}-1)(q^{12}-1)(q^{8}-1)(q^{2}-1) }{( q^{120} (q-1)^8 )^2 }\\
            &=  \frac{(q^{8}-1)}{(q-1)^{16} } \prod_{i\in \{2,12,14,18,20,24,30 \}} (1- \frac{1}{q^{i}}) \\
            & \leq  \frac{q^{8}}{(q-1)^{16} } \\
            & \leq 1.  \quad \quad \quad  \quad \quad \quad  \quad \quad \quad ( \text{by Remark}\, \ref{Remark-1})
        \end{align*}
This implies that,
$$ \Big (\frac{\lvert H \rvert}{\lvert  H_1 \rvert} \Big)^2 \leq \lvert H \rvert.$$
   Notice that, the group $E_8(2)$ is of constant order. However, we can use Method 2 to reduce the constants $b_1$, $b_2$ to $1$. By taking $H_1$ to be maximal subgroup of order (see \cite{ATLAS}) $2^{119}\cdot3^{4}\cdot 5\cdot7^{2}\cdot31$ and $H_2$ to be its Sylow $2$-subgroup of order $2^{119}$.
    \vspace{0.1cm}
    
    \item[(8)]  $H =~^2B_2(q)$, $q=2^{2t+1}$ and $t\geq1$ (Method 1)

\vspace{0.1cm}
    The finite simple group $^2B_2(q)$ has order (see \cite{ASH}, p. 252), $$\vert H \vert = q^2(q^2+1)(q-1).$$
    Thus, it has a Sylow 2-subgroup  $H_2$  of order $q^2$ and $ \lvert H_2 \rvert^2  \leq \lvert H \rvert.$ Let $H_1$ be the Borel subgroup of $^2B_2(q)$ then the order $H_1$ is  $q^{2}(q-1)$ (see \cite{WLS}, p.  117).
    
    Consider,
    \begin{align*}
    \frac{\lvert  H_1 \rvert}{\lvert H_2 \rvert^2}
            &=\Big(\frac{q^{2}(q-1)}{q^{4}}\Big)\\
            & < 1.
    \end{align*}
    Therefore,
    $$\Big (\frac{\lvert  H_1 \rvert}{\lvert H_2 \rvert} \Big )^2  < \lvert  H_1 \rvert < \lvert H \rvert.$$
    Now, 
     \begin{align*}
         \frac{\lvert H \rvert}{\vert H_1 \vert^2} 
            &=\frac{q^{2}(q^2+1)(q-1)}{(q^{2}(q-1))^2}\\
            & = \frac{(q^2+1)}{q^{2}(q-1)}\\
            & = \frac{1+\frac{1}{q^2}}{q-1}\\
            & < 1.
    \end{align*}
Thus, $$ \Big (\frac{\lvert H \rvert}{\lvert  H_1 \rvert} \Big )^2 <  \vert H  \vert. $$

     \item[(9)] $H= ~^2G_2(q)$; where, $q=3^{2t+1}$ and $t\geq 1$ (Method 1)
        
        \vspace{0.1cm}
        The group $^2G_2(q)$ has order (see \cite{ASH}, p. 252), $$\lvert  ~^2G_2(q) \rvert = q^3(q^3+1)(q-1).$$ Thus, $^2G_2(q)$ has a Sylow 3-subgroup of order $q^3$  and  $\lvert H_2 \rvert^2 < \lvert H \rvert.$ Also, $ ^2G_2(q)$ has the Borel subgroup $H_1$ of order $q^3(q-1)$ (see\cite{WLS}, p. 137).
        
        Consider,
        \begin{align*}
        \frac{\lvert  H_1 \rvert}{\lvert H_2 \rvert ^2}
            &=\frac{q^{3}(q-1)}{q^{6}}\\
            &< 1.
        \end{align*}
Therefore,
$$ \Big (\frac{\lvert  H_1 \rvert}{\lvert H_2 \rvert} \Big)^2  < \lvert  H_1 \rvert < \lvert H \rvert.$$
         Now,
        \begin{align*}
        \frac{\lvert H \rvert}{ \lvert  H_1  \rvert ^2}
            &=\frac{q^3(q^3+1)(q-1)}{(q^{3}(q-1))^2}\\
            &=\frac{(q^3+1)}{q^{3}(q-1)}\\
            &= \frac{1+\frac{1}{q^3}}{q-1}\\
            & < 1.
            \end{align*}
Therefore,
$$ \Big(\frac{\lvert H \rvert}{\lvert  H_1 \rvert} \Big)^2 < \lvert H \rvert.$$
    
     \item [(10)] $H= ~^2F_4(q)$, $q=2^{2t+1}$ and $t\geq1$ (Method 1)
        
        \vspace{0.1cm}
        The order of the group $^2F_4(q)$ is $q^{12}(q^6 +1)(q^4-1)(q^3+1)(q-1)$ (see \cite{ASH}, p. 252). Thus, it has a Sylow $2$-subgroup of order $q^{12}$ and $\lvert H_2 \rvert^2 < \vert H \vert.$
        Also, $^2F_4(q)$ has the Borel subgroup $H_1$ of order $q^{12}(q-1)^2$ (see \cite{WLS}, p.  165).
        
        Consider,
        \begin{align*}
        \frac{\lvert  H_1 \rvert}{ \lvert H_2 \rvert ^2}
            &=\frac{q^{12}(q-1)^2}{q^{24}}\\
            & < 1.
        \end{align*}
    Therefore,  $$ \Big (\frac{\lvert  H_1 \rvert}{\lvert H_2 \rvert} \Big)^2  < \lvert  H_1 \rvert < \lvert H \rvert.$$
      \begin{align*}
        \frac{\lvert H \rvert}{ \lvert  H_1 \rvert ^2}
            &=\frac{q^{12}(q^6+1)(q^4-1)(q^3+1)(q-1)}{(q^{12}(q-1)^2)^2}\\
            &= \frac{(q^6+1)(q^4-1)(q^3+1)(q-1)}{q^{12}(q-1)^4}\\
            &=\frac{(q^6+1)(q^4-1)(q^3+1)}{q^{12}(q-1)^3}\\
            & < \frac{q^4 (1+\frac{1}{q^6})(1+\frac{1}{q^3})}{q^3(q-1)^3}\\
            &= \frac{q (1+\frac{1}{q^6})(1+\frac{1}{q^3})}{(q-1)^3}\\
            & <1.
        \end{align*}
        This implies that,
        $$  \Big (\frac{\vert H) \vert}{\lvert  H_1 \rvert} \Big)^2 < \lvert H \rvert.$$
    
    \item[(11)] $H= ~^2F_4(2)';$ (Method 2)
    
The simple group $H= ~^2F_4(2)';$ has order $17971200.$ It is known that $H$ has a maximal subgroup of order $11232$. We take $H_1$ to be this maximal subgroup and $H_2$ to be the Sylow $2$-subgroup of $H_1$ which has order $32$. Thus, we get $b_{1}=b_{2}=1$. 
\end{enumerate}

\subsection{Tables}\label{tables}

In this section we cover the details of Sporadic simple groups (Table \ref{Table-2}), and the order of all the simple groups that we define in cases (ii)-(iv) of Theorem \ref{classification-thm} in Table~\ref{Table-3} and \ref{Table-4}. These tables also contain the order of the subgroups $H_2$ and $H_1$. Table \ref{Table-2} represents the information about the subgroups $H_2$ and $H_1$ of the Sporadic simple groups.

\newpage
\begin{table}
\centering
\begin{tabular}{|p{.85cm}| p{4.5cm} | p{1.3cm}| p{3.9cm}| p{0.5cm}| p{0.5cm}|}
\hline
$H$ & Order of $H$ & Order of $H_2$ &  Order of $H_1$ &  $ b_1 $ & $ b_2 $ \\ \hline

 ${\rm M_{11}}$ & 7920 & $2^4$ & 720 &   $1$ & $1$  \\  \cline{1-6} 
 ${\rm M_{12}}$  & 95040 & $2^2$ & 660 &   $1$ & $1$   \\  \cline{1-6}
 ${\rm M_{22}}$  & 443520 & $2^6$ & 20160 &   $1$ & $1$   \\  \cline{1-6}
 ${\rm M_{23}}$  & 10200960 & $ 2^7$ & 443520 &   $1$ & $1$   \\  \cline{1-6}
 ${\rm M_{24}}$  & 244823040 & $2^8 $ & 887040 &   $1$ & $1$   \\  \cline{1-6}
 
  ${\rm {Co}_{1}}$ & 4157776806543360000 & 262144 & 42305400000000 &   $1$ & $1$   \\  \cline{1-6}
  ${\rm {Co}_{2}}$ & 42305400000000 & 262144 & 908328960 &   $1$ & $1$   \\  \cline{1-6}
  ${\rm {Co}_{3}}$ & 495767000000 & $ 2^7$ & 10200960 &   $1$ & $1$   \\  \cline{1-6}
  ${\rm McL}$ & $898128000$ & $3^{6}$ & $3^6 \cdot 2^7 \cdot7 \cdot 5 $ &   $1$ & $1$   \\  \cline{1-6}
  ${\rm HS}$ &  $44352000$ & $2^{7}$ & $2^{7} \cdot 3^{2} \cdot 5\cdot7\cdot11 $ &   $1$ & $1$   \\  \cline{1-6}
 
  ${\rm Suz}$ & $448345497600$ & $2^{12}$ & $251596800$ &   $1$ & $1$   \\  \cline{1-6}
  ${\rm J_{2}}$ &  604800 & $2^5$ & 6048 &   $1$ & $1$   \\  \cline{1-6}
 
  ${\rm Fi_{22}}$ & $ 64561751654400 $ & $2^{16}$ & $2^{16}(2^{6}-1)(2^{5}+1)(2^{4}-1)(2^{3}+1)$ &   $1$ & $1$   \\  \cline{1-6}
  ${\rm Fi_{23}}$ & $4089470473293004800$  &  $2^{18}$ & $2^{18} \cdot 3^{9}\cdot 5^{2}\cdot 7 \cdot 11\cdot 13$ &   $1$ & $1$   \\  \cline{1-6}
 
  ${\rm Fi^{`}_{24}}$ & 1255205709190661721292800 & $2^{19}$ & $2^{19} \cdot 3^{13} \cdot 5^{2} \cdot 7 \cdot 11 \cdot 13 \cdot 17 \cdot 23$  &   $1$ & $1$   \\  \cline{1-6}
 
  ${\rm \mathbb{M}}$ & $t_1$ & $ 2^{42}$ & $ t_2 $ &   $1$ & $1$   \\  \cline{1-6}
  ${\rm \mathbb{B}}$ & $t_3$ & 2$^{38}$ & $t_4 $ &   $1$ & $1$   \\  \cline{1-6}
  ${\rm Th}$ & $90745943887872000$ & $2^{15} $ & $319979520$ &   $1$ & $1$   \\  \cline{1-6}
  ${\rm HN}$ & $273030912000000$ & $2^9$ & $239500800$ &   $1$ & $1$   \\  \cline{1-6}
  ${\rm He}$ & $4030387200$ & $2^{8}$ & $2^{8} \cdot 255 \cdot 15$ &   $1$ & $1$   \\  \cline{1-6}
 
  ${\rm J_1}$ & 175560 & $2^2$ & 660 &   $1$ & $1$   \\  \cline{1-6}
  ${\rm J_3}$ & 50232960 & $2^5$ & 8160 &   $1$ & $1$   \\  \cline{1-6}
  ${\rm J_4}$ & 86775571046077562880 & 2097152 & 57161637225 &   $1$ & $1$   \\  \cline{1-6}
  ${\rm O'N}$ & $460815505920$ & $2^6$ & $3753792$ &   $1$ & $1$   \\  \cline{1-6}
  ${\rm Ly}$ & $51765179004000000$ & $ 15625 $ & $5859000000$ &   $1$ & $5$   \\  \cline{1-6}
  ${\rm Ru}$ & $145926144000$ & $2^{12}$ & $35942400 $  &   $1$ & $1$ \\ \cline{1-6}
\end{tabular}

\caption{Table representing the constant factor and Method used for choosing suitable subgroups}
\label{Table-2}
\end{table}

\noindent In Table \ref{Table-2} we consider the values of $t_i, i=1,2,3,4$ as follows. 

$t_1=808017424794512875886459904961710757005754368000000000$

$t_2= 2^{42} \cdot 3^{13} \cdot 5^{6} \cdot 7^{2} \cdot 11 \cdot 13 \cdot 17 \cdot 19 \cdot 23 \cdot 31 \cdot 47$

$t_3= 4154781481226426191177580544000000 $

$t_4= 2^{38}\cdot (2^{12}-1) \cdot (2^{9}+1) \cdot (2^{8}-1)\cdot (2^{6}-1)\cdot (2^{5}+1)\cdot (2^{2}-1)$.

\noindent In the Table \ref{Table-3}, the values of $c_{mi}, i=1,2,3,4,5$ are as follows. 

$c_{m1}=\frac{q^{\frac{m(m+1)}{2} }} {(q+1,m+1)} $ $(q-1)^{\lfloor m/2 \rfloor} (q+1)^{\lceil \frac{m-1}{2}\rceil}$, $c_{m2}= \frac{3}{(3,m+1)} \frac{ {2^{\frac{m(m+1)}{2}}}  \prod_{i=1}^{m} (2^{i+1}-(-1)^{i+1})} {(2^{m+1}- (-1)^{m+1})(2^{m}-(-1)^{m})}$

$c_{m3} = \frac{3}{(3,m+1)}  {2^{\frac{m(m-1)}{2}}}  \prod_{i=1}^{m-1} (2^{i+1}-(-1)^{i+1})$, $c_{m4}=2^{m^{2} -m+1} (2^{m} +1)  \prod_{i=1}^{m-1} (2^{2i}-1)$

$c_{m5}= 2^{m(m-1)} (2^{m-1}+1)\prod_{i=1}^{m-2} (2^{2i}-1)$.

 \begin{table}
\centering
 
 %\caption{Table 3}
\begin{tabular}{|p{2.1cm}| p{4.9cm} | p{1.6cm}| p{3.5cm}|} 
 \hline
 $H$ & $\lvert H \rvert$ & $\lvert H_2 \rvert$ & $\lvert H_1 \rvert $\\ 
\hline
 $A_{m}(q)$; $q>2$ & $\frac{q^{\frac{m(m+1)}{2} }  \prod_{i=1}^{m} (q^{i+1}-1)  }{(q-1,m+1)}$&  $q^{\frac{m(m+1)}{2}}$  & $\frac{q^{\frac{m(m+1)}{2}}}{(q-1,m+1)} (q-1)^m$ \\ 
  \hline
$A_{m}(2)$; $q=2$ & $2^{\frac{m(m+1)}{2} }  \prod_{i=1}^{m} (2^{i+1}-1) $ & $ 2^{\frac{m(m+1)}{2} }$  & $\frac{ 2^{\frac{m(m+1)}{2} }  \prod_{i=1}^{m} (2^{i+1}-1) }{ (2^{m+1}-1) }$\\ 
 \hline
%----
$^2A_m(q^2)$; $q>2$, $m >1$ & $\frac{q^{\frac{m(m+1)}{2} } }{(q+1,m+1)}\prod_{i=1}^{m}(q^{i+1}-(-1)^{i+1})$ & $q^{\frac{m(m+1)}{2} }$  &  $c_{m1} $ \\

\hline
 
 $^2A_m(2^2)$; $q= 2$, $m>1$, $6 \nmid (m-1)$ & $\frac{2^{\frac{m(m+1)}{2} } }{(3,m+1)}  \prod_{i=1}^{m} (2^{i+1}-(-1)^{i+1})$ & $2^{\frac{m(m+1)}{2}}$ & $c_{m2}$ \\
 \hline
 $^2A_m(2^2)$; $q= 2$, $m>1$, $6 \mid (m-1) $ & $\frac{2^{\frac{m(m+1)}{2} } }{(3,m+1)}  \prod_{i=1}^{m} (2^{i+1}-(-1)^{i+1})$ & $2^{\frac{m(m-1)}{2}}$ & $c_{m3}$ \\ 
  \hline
  
$C_{m}(q)$; $q>2$, $m>2 $ & $\frac{q^{m^2 }  \prod_{i=1}^{m} (q^{2i}-1)  }{(2,q-1)}$ &  $q^{m^2}$  & $\frac{q^{m^2}}{(2,q-1)} (q-1)^m$  \\ 
 \hline
 $C_{m}(2)$; $q=2$, $m>2 $ &  $2^{m^2 }  \prod_{i=1}^{m} (2^{2i}-1)$ & $ 2^{m^{2} -m+1} $  & $c_{m4} $ \\ 

 \hline
$B_{m}(q)$; $q$ odd, $m>1$ & $\frac{q^{m^2 }  \prod_{i=1}^{m} (q^{2i}-1)  }{(2,q-1)}$ & $q^{m^2}$  & $\frac{q^{m^2}}{(2,q-1)} (q-1)^m$ \\ 

\hline
 $D_{m}(q)$; $q>2$, $m>3 $ & $\frac{q^{{m(m-1)}} (q^{m}-1) \prod_{i=1}^{m-1} (q^{2i}-1)  }{(4, q^{m}-1)}$ &  $q^{{m(m-1)}}$ & $\frac{q^{{m(m-1)}} }{ (4, q^{m}-1)}  (q-1)^m$ \\ 

 \hline
$D_{m}(2)$; $q=2$, \newline $m>3 $ & $2^{{m(m-1)}} (2^{m}-1) \prod_{i=1}^{m-1} (2^{2i}-1)$ &  $ 2^{m^{2}-2m+1}  $ & $ 2^{m^{2}-2m+1} \prod_{i=1}^{m-1} (2^{2i}-1)$ \\  
 \hline
$^2D_m(q^2)$;\newline $q> 2$, $m>3$ & $\frac{q^{m(m-1)}(q^{m}+1) }{(4, q^{m}+1)} \prod_{i=1}^{m-1} (q^{2i}-1)  $ & $q^{{m(m-1)}}$ & $\frac{q^{{m(m-1)}} }{ (4, q^{n}+1)}  (q-1)^m $ \\
 \hline
 $^2D_m(2^2)$; $q= 2$, $m>3$ & $\frac{2^{m(m-1)}(2^{m}+1) }{(4, 2^{m}+1)} \prod_{i=1}^{m-1} (2^{2i}-1)$ &  $2^{m(m-1)}$ & $c_{m5} $\\ 
 \hline
 \end{tabular}

    \caption{Order of the simple groups (case (iii) of Theorem \ref{classification-thm}) and order of its subgroups $H_2, H_1$. }
    \label{Table-3}
\end{table}

 \begin{table}
\centering
 
 %\caption{Table 3}
\begin{tabular}{|p{2.3cm}| p{4.9cm} | p{1.3cm}| p{3.3cm}|} 
 \hline
 $G_2(q)$ & $q^{6} (q^{6}-1)(q^{2}-1)$ &  $q^{6}$  & $q^{6}(q-1)^2$ \\
 
\hline
 $F_4(q)$ &$q^{24} \prod_{i \in \{2,6,8,12\}} (q^i -1)$ & $q^{24}$ & $q^{24} \prod_{i \in \{1,2,4,6\}} (q^i -1)$ \\ 
\hline 
$E_6(q)$, $q>2$ & $\frac{q^{36}}{(3,q-1)} \prod_{i \in \{2,5,6,8,9,12\}} (q^i -1)$ &$q^{36}$ & $q^{36} (q-1)^6$ \\
\hline
 $E_6(2)$ & $2^{36} \prod_{i \in \{2,5,6,8,9,12\}} (2^i -1)$ & $2^{36}$ & $2^{36}\cdot3^{3}\cdot5\cdot7\cdot31$ \\
\hline
$^2E_6(q)$  & $\frac{q^{36} (q^9+1)}{(3,q+1)} \prod_{i \in \{2,5,6,8,12\}} (q^i -1)$ & $q^{36}$ & $q^{36} (q-1)^4(q+1)^2 $ \\ 
\hline
$^3D_4(q)$ & $ q^{12} (q^{8}+q^4+1)(q^{6}-1)(q^{2}-1)$ & $q^{12}$ &  $q^{12} $ $(q^3-1) (q-1)$ \\ 
 \hline 
$E_7(q)$, $q>2$ & $\frac{q^{63}}{(2,q-1)} \prod_{i \in \{2,6,8,10,12,14,18\}} (q^i -1)$ &  $q^{63}$ & $q^{63} (q-1)^7$ \\ 
 \hline
 $E_7(2)$ & $2^{63} \prod_{i \in \{2,6,8,10,12,14,18\}} (2^i -1)$ &  $2^{63}$ & $2^{63}\cdot3^{4}\cdot7^{2}\cdot5$ \\
 \hline
$E_8(q)$, $q>2$ & $ q^{120} \prod_{i \in \{2,8,12,14,18,20,24,30\}} (q^i -1)$ &  $q^{120} $ & $q^{120} (q-1)^8$  \\ 
 \hline
 $E_8(2)$ & $ 2^{120}  \prod_{i \in \{2,8,12,14,18,20,24,30\}} (2^i -1)$ &  $2^{119} $ & $2^{119}\cdot3^{4}\cdot 5\cdot7^{2}\cdot31$  \\ 
 \hline
 $^2B_2(q)$;\newline $q=2^{2t+1},$ $t \geq 1 $ & $ q^2(q^2+1)(q-1)$ &  $q^{2}$  & $q^{2}(q-1)$\\ 
 \hline
 
$^2G_2(q)$;\newline $q=3^{2t+1},$ $t \geq 1$  & $q^3(q^3+1)(q-1)$ & $q^3$ & $q^3(q-1)$\\
 \hline
$^2F_4(q)$;\newline $q=2^{2t+1},$ $t \geq 1 $ & $ q^{12}(q^6 +1)(q^4-1)(q^3+1)(q-1)$ &  $q^{12}$ & $q^{12}(q-1)^2$\\ 
 \hline

\end{tabular}

    \caption{Order of the simple groups (case (iv) of Theorem \ref{classification-thm}) and order of its subgroups $H_2, H_1$.}
    \label{Table-4}
\end{table}

\end{document}